\documentclass[aps,prx,reprint,superscriptaddress]{revtex4-2}
\pdfoutput=1
\usepackage[T1]{fontenc}
\usepackage[utf8]{inputenc}
\usepackage{microtype}
\usepackage{comment}
\usepackage{amsmath}
\usepackage{amssymb}
\usepackage{amsthm}
\usepackage{graphicx}
\usepackage{mathrsfs}
\usepackage{braket}
\usepackage{bbm}
\usepackage{xcolor}
\usepackage{tikz}
\usepackage{booktabs}
\usepackage{threeparttable}
\usepackage{multirow}

\usepackage{theoremref}

\usepackage{algorithm}
\usepackage{algorithmic}

\newtheorem{thm}{\protect\theoremname}
\theoremstyle{plain}
\newtheorem{lem}{\protect\lemmaname}
\theoremstyle{plain}

\theoremstyle{plain}
\newtheorem*{lem*}{\protect\lemmaname}
\theoremstyle{plain}
\newtheorem*{thm*}{\protect\theoremname}
\theoremstyle{plain}

\theoremstyle{plain}
\newtheorem{cor}{\protect\corollaryname}
\theoremstyle{plain}
\newtheorem*{cor*}{\protect\corollaryname}

\newtheorem{defn}{Definition}
\newtheorem*{defn*}{Definition}

\usepackage[USenglish]{babel}
  \providecommand{\corollaryname}{Corollary}
  \providecommand{\lemmaname}{Lemma}
  \providecommand{\propositionname}{Proposition}
  \providecommand{\remarkname}{Remark}
\providecommand{\theoremname}{Theorem}

\usepackage[normalem]{ulem}
\usepackage{hyperref}
\hypersetup{colorlinks=true, linkcolor=blue, urlcolor=blue, citecolor = blue}

\newcommand{\wt}{\widetilde}

\newcommand{\dd}{\mathrm{d}}
\newcommand{\Tr}{\mathrm{Tr}}
\newcommand{\norm}[1]{\lVert #1 \rVert}
\newcommand{\Or}{\mathcal{O}}

\DeclareMathOperator{\spanop}{span}

\usepackage{todonotes}

\newcommand{\mathplaceholder}{\,\cdot\,}
\newcommand{\intinf}{\int_{-\infty}^\infty}

\newcommand{\suppm}{\mathrm{supp}}

\newcommand{\calA}{\ensuremath{\mathcal{A}}}
\newcommand{\calL}{\ensuremath{\mathcal{L}}}
\newcommand{\sfH}{\ensuremath{\mathsf{H}}}
\newcommand{\calO}{\ensuremath{\mathcal{O}}}
\newcommand{\calB}{\ensuremath{\mathcal{B}}}
\newcommand{\calH}{\ensuremath{\mathcal{H}}}
\newcommand{\frakH}{\ensuremath{\mathfrak{H}}}
\newcommand{\bbR}{\ensuremath{\mathbb{R}}}
\newcommand{\bbC}{\ensuremath{\mathbb{C}}}

\begin{document}

\title{Quantum Gibbs sampling through the detectability lemma}

\newcommand{\dukemath}{Department of Mathematics, Duke University}
\newcommand{\dqc}{Duke Quantum Center, Duke University}
\newcommand{\dukephy}{Department of Physics, Duke University}
\newcommand{\dukechem}{Department of Chemistry, Duke University}
\newcommand{\dukeece}{Department of Electrical and Computer Engineering, Duke University}

\author{Di Fang}
\affiliation{\dukemath}
\affiliation{\dqc}

\author{Jianfeng Lu}
\affiliation{\dukemath}
\affiliation{\dukephy}
\affiliation{\dukechem}

\author{Yu Tong}
\affiliation{\dukemath}
\affiliation{\dqc}
\affiliation{\dukeece}

\author{Chu Zhao}
\affiliation{\dukeece}
\date{\today}

\begin{abstract}
    Gibbs state preparation
    is an important subroutine in quantum computing. In this work we use the detectability lemma to improve Gibbs state preparation. Specifically, we design new Gibbs state preparation methods that do not rely on simulating Lindbladian evolution, thus avoiding the overhead from it. For local Lindbladians consisting of $M$ terms, this approach reduces the cost by a factor of $\Or(M)$. We also combine the detectability lemma operator and quantum singular value transformation to implement ground state projection operators of frustration-free Hamiltonians, resulting in a quadratic speedup in the spectral gap dependence. Applying this method to Lindbladians for the Gibbs state of local commuting Hamiltonians, we achieve quadratically better dependence on the Lindbladian spectral gap.
\end{abstract}

\maketitle

\section{Introduction}

Efficiently sampling from the Boltzmann distribution is a cornerstone problem in classical statistical physics and machine learning. Its quantum counterpart, the preparation of the quantum Gibbs state, is also central to the study of quantum phase transitions and the process of thermalization. Moreover, it serves as a subroutine for various quantum algorithms, ranging from quantum semi-definite programming \cite{BrandaoKalev2017quantum, vanApeldoorn2017quantum} to simulating condensed matter systems at finite temperature \cite{Motta2020determining,chen2023quantumthermalstatepreparation,LloydAbanin2025quantum}.

In recent years, significant advances have been made in engineering Lindbladian dynamics to prepare Gibbs states~\cite{chen2023quantumthermalstatepreparation,ChenKastoryanoBrandaoGilyen2023,RallWangWocjan2023,DingLiLin2024efficient}.
The Lindbladian generators constructed in these works can be efficiently implemented on quantum computers using state-of-the-art Lindbladian solvers, such as those in~\cite{CleveWang2017, WocjanTemme2023, ChenKastoryanoBrandaoGilyen2023, ChenKastoryanoGilyen2023, LiWang2023, DingLiLin2024simulating,PocrnicSegalWiebe2024,chen2025efficient}, leading to efficient quantum algorithms for Gibbs state preparation if the Lindbladian dynamics converges quickly to the Gibbs state. In this approach, the overall complexity depends crucially on the mixing time of the Lindbladian dynamics. Consequently, much research has focused on establishing rigorous mixing time bounds for various physical systems \cite{temme2014hypercontractivity,ding2024polynomial,bardet2024entropy,alicki2009thermalization,ramkumar2024mixing,KochanowskiAlhambraCapelRouze2024rapid,BardetCapelLiEtAl2023rapid,RouzeFrancaAlhambra2024optimal,capel2024quasi,kastoryano2013quantum,temme2015fast,ChenLiLuYing2024randomized,Rakovszky2024bottlenecks,LiLu2024quantum,Barthel2022solving,Fang2024mixing,DingLiLin2024efficient,DingChenLin2024single,TongZhan2025fast,BakshiLiuMoitraTang2025dobrushin}.

However, constructing a Lindbladian dynamics and then simulating it on a quantum computer may not be the most efficient way to prepare Gibbs states. Simulation requires the quantum state produced by the algorithm to follow the dynamics closely along the entire trajectory of the dynamics. This requirement is overly restrictive for Gibbs state preparation: it is sufficient for the dynamics produced by the quantum algorithm to converge to a state that is close to the Gibbs state on a timescale comparable to the original dynamics. One algorithm following this idea is proposed in \cite{BakshiLiuMoitraTang2025dobrushin}, which randomly applies a quantum channel corresponding to a Lindbladian term at each step, thus avoiding the overhead from accurately simulating the Lindbladian. Similar ideas have been used to improve Lindbladian simulation \cite{ChenLiLuYing2024randomized}. {Discrete-time quantum processes satisfying an exact detailed balance condition without simulating Lindbladian evolution have also been developed \cite{GilyenChenDoriguelloKastoryano2024quantum}. A related quantum Metropolis scheme implements a quantum channel for which the Gibbs state is an approximate fixed point \cite{JiangIrani2024quantum}.}
We see the same consideration in classical Markov Chain Monte Carlo methods, where the objective is the design of an efficient transition kernel that ensures rapid convergence to the stationary distribution, rather than reproducing a specific physical stochastic process.

Another way to improve the runtime of the quantum algorithm is through improving the dependence on the spectral gap. While directly simulating the Lindbladian typically incurs a cost that scales linearly in inverse spectral gap, it is sometimes possible to design algorithms that run in time that scales as the \textit{square root} of the inverse spectral gap, thus achieving a quadratic speedup. Examples include the Lindbladian constructed in \cite{chen2023quantumthermalstatepreparation}, and certain Lindbladians corresponding to Hamiltonians whose eigenvalues can either be exactly computed or satisfy a ``rounding promise'' \cite{WocjanTemme2023}. These works follow the idea developed for quantizing a classical reversible Markov chain in \cite{szegedy2004quantum}, which also resulted in a quadratic improvement in the gap dependence. A recent result also extended such speedup to the continuous-time setting for classical Langevin dynamics \cite{LengDingChenLin2026operator}.

Detectability lemma (DL) is a powerful technique to study Hamiltonian complexity and to prove entanglement area laws \cite{AharonovAradVaziraniLandau2011detectability,AradKitaevLandauVazirani2013area,AharonovAradLandauVazirani2009}.
In the detectability lemma, the DL operator denotes the layered product of local ground-space projectors, which acts as an approximate projector onto the global ground space.
In this work, we use the DL operator as an algorithmic tool to improve Gibbs state preparation along both of these directions. More specifically, we will construct efficiently implementable quantum channels that bypass Lindbladian simulation,
thereby achieving a speedup by a factor of $M$, the number of local Lindbladian terms. These channels update the quantum state locally and are similar in spirit to classical Gibbs sampling algorithms \cite{ascolani2024entropy, liu1995covariance} and coordinate hit-and-run methods \cite{narayanan2022mixing}, where subsets of variables are updated conditional on the rest.
We will also consider the parent Hamiltonian of the Lindbladian, and use the detectability lemma (DL)
to construct an efficiently implementable projection operator for its ground state. This allows us to achieve quadratically improved dependence on the spectral gap when an annealing path is available. 

\subsection{Main results}

Our first main result is an algorithm to prepare the stationary state of a Lindbladian without simulating its time evolution. When the Lindbladian consists of $M$ local terms, this approach achieves a linear dependence on $M$ in its runtime. By comparison, state-of-the-art Lindbladian simulation algorithms, such as \cite{DingLiLin2024simulating}, need to first normalize the Lindbladian, thus resulting in an $M$ factor appearing in the rescaled evolution time, making the simulation cost depending quadratically on $M$. Our method therefore achieves a factor of $M$ speedup compared to the simulation-based method. Below we provide a precise statement of the result:
\begin{thm}[Corollary~\ref{cor:gate_complexity}]
\label{thm:gate_complexity_main}
    Let $\mathcal{L}=\sum_{m=1}^M \mathcal{L}_m$, where each $\mathcal{L}_m$ is a Lindbladian that satisfies $\sigma$-KMS detailed balance condition, and we assume $\mathcal{L}$ has a non-degenerate $0$-eigenspace. We assume that each $\mathcal{L}_m$ commutes with all but at most $g$ others. Then we can prepare $\sigma$ up to $\epsilon$ error in trace distance using $\wt{\Or}\left(\frac{M g^2}{\mathrm{gap}(\mathcal{L})}\log\left(\frac{1}{\sigma_{\min}\epsilon}\right)\log^c\left(\frac{M}{\epsilon}\right)\right)$ elementary single- and two-qubit gates, where $\sigma_{\min}$ is the minimum eigenvalue of $\sigma$, and $c$ is a constant from the Solovay-Kitaev theorem.
\end{thm}
In the above, the KMS detailed balance condition is defined in Definition~\ref{defn:KMS_detailed_balance}. $\mathcal{L}$ satisfying $\sigma$-KMS detailed balance condition implies that $\sigma$ is its stationary state. The algorithm consists of a sequence of local quantum channels that are applied repeatedly. The convergence towards the stationary state $\sigma$ is analyzed using the detectability lemma (Lemma~\ref{lem:detectability_lemma}).

For the second direction of refinement, the DL operator can be used to obtain a quadratic improvement in the spectral gap dependence. For this, we require the parent Hamiltonian of the Lindbladian to have strict locality. To ensure this is true, we focus on preparing the Gibbs states of commuting local Hamiltonians. 

\begin{thm}[Corollary~\ref{cor:commuting}]
\label{thm:commuting_main}
    Consider commuting and bounded-degree local Hamiltonian $H$. For each $\nu\in [0, \beta]$, we construct $\sigma_\nu$-detailed balanced $\calL_\nu$ as stated in Theorem~\ref{thm:gibbsstateprepcomplex}
    and assume its irreducibility. Then, the corresponding parent Hamiltonian $\sfH_\nu$ will be bounded-degree local and frustration-free, thus it satisfies the requirements of Theorem~\ref{thm:gibbsstateprepcomplex}. Therefore there exists a quantum algorithm that prepares the purified Gibbs state at inverse temperature $\beta$ with error $\delta$ and success probability at least $1-\delta$ with gate complexity
    \begin{equation}
        \calO \left( \frac{M\beta \norm{H}}{\sqrt{\gamma}} \log^2 \left(\frac{\beta\norm{H}}{\delta}\right) \log^c \left( \frac{M}{\sqrt{\gamma}} \frac{\beta\norm{H}}{\delta} 
        \right) \right)
    \end{equation} using $\calO (\log M) + 1$ resettable ancilla qubits, where $c$ is the exponent in the Solovay-Kitaev theorem, and $\gamma = \min_j \mathrm{gap}(\mathcal{L}_{\beta_j})$, $\beta_j = j\beta/K$, $K=\Theta(\beta\|H\|)$.
\end{thm}

To prove the above theorem, we propose an algorithm for the frustration-free Hamiltonian ground state problem based on the DL operator, which may be of independent interest. By applying quantum singular value transformation (QSVT) \cite{Gilyen_2019singularvalue} to the DL operator, we obtain a ground state projection operator whose cost scales as the inverse square root of the Hamiltonian spectral gap. Compared to previous methods that uses the block encoding of the original Hamiltonian to prepare the ground state \cite{LinTong2020,Ge2018fastergroundstatepreparation}, this approach achieves a quadratic speedup in the spectral gap dependence.
\begin{thm}[Corollary~\ref{cor:gate_complexity_approximate_proj}]
    \label{cor:gate_complexity_approximate_proj_main}
    Let $H = \sum_{m=1}^M H_m$, where each $H_m$ satisfies $0\preceq H_m\preceq I$, is supported on at most $\mathsf{k}$ qubits, and overlaps with at most $g$ other terms $H_{m'}$. We assume that $H$ is frustration-free in the sense that $P_H H_m=0$ for all $m=1,2,\cdots,M$, where $P_H$ is its ground space projection operator. Then a $(1,\Or(\log(M)),\epsilon)$-block encoding of $P_H$ (see Definition~\ref{defn:block_encoding}) can be obtained using 
    \[
    \Or\left(M\gamma^{-1/2}\log(1/\epsilon)\log^c(M\gamma^{-1/2}\log(1/\epsilon)/\epsilon)\right)
    \]
    elementary gates, where $c$ is the exponent in the Solovay-Kitaev theorem.
\end{thm}

\subsection{\texorpdfstring{{Algorithmic overview}}{Algorithmic overview}}
\label{sec:algorithmic_overview}

{We briefly summarize the two algorithms that we develop in this work. The first algorithm takes as input a local KMS-detailed-balanced Lindbladian $\mathcal{L}=\sum_{m=1}^M\mathcal{L}_m$. For every term, it implements the local channel $P_m^\dag$, where $P_m=\lim_{t\to\infty}e^{t\mathcal{L}_m}$, and applies one sweep $\Phi^\dag=P_M^\dag\cdots P_1^\dag$. Repeating this sweep makes the quantum system converge to the stationary state. The detectability lemma controls the contraction per sweep in terms of $\mathrm{gap}(\mathcal{L})$ and the interaction degree $g$, while the circuit directly implements local channels and never simulates the trajectory $e^{t\mathcal{L}^\dag}$.}

{The second algorithm prepares a purified Gibbs state along a temperature schedule $0=\beta_0<\beta_1<\cdots<\beta_K=\beta$, with a properly chosen number of steps $K$. At each temperature, the KMS-detailed-balanced Lindbladian is mapped to a frustration-free parent Hamiltonian whose ground state is the purified Gibbs state. For commuting bounded-degree local Hamiltonians, this parent Hamiltonian is itself bounded-degree local. We compose block encodings of its local ground-space projectors to obtain a block encoding of the DL operator and then use QSVT to amplify the singular-value separation supplied by the detectability lemma. This produces an approximate global ground-space projector with inverse-square-root gap dependence. A second QSVT step applied to two consecutive approximate projectors implements the transition $\ket{\psi_{\beta_{j-1}}}\mapsto\ket{\psi_{\beta_j}}$. Annealing is then done by composing these transition operators.}

\subsection{Organization}
The remainder of this paper is organized as follows. In Section \ref{sec:problem_setup}, we define the notation and provide the formal definitions for KMS detailed balance. Section \ref{sec:detectability_lemma} reviews the detectability lemma and applies it to Lindbladians. In Section \ref{sec:implementing_ground_state_projection_operators}, we describe the implementation of Hamiltonian ground state projection operators via the detectability lemma with a quadratic speedup in the spectral gap dependence. Section \ref{sec:parent_hamiltonian} details the construction of the parent Hamiltonian of a KMS-detailed balanced Lindbladian. Section \ref{sec:annealing} presents the state preparation algorithm via annealing with quadratically improved spectral gap dependence. Finally, in Section \ref{sec:commuting_ham}, we apply these techniques to commuting bounded-degree local Hamiltonians and discuss the resulting gate complexity.

\section{Lindbladians and detailed balance}
\label{sec:problem_setup}

We consider an $n$-qubit quantum system, with a Hilbert space $\mathcal{H}$. The operator algebra is denoted by $B(\mathcal{H})$ while the superoperator algebra (consisting of linear operators on the operator algebra) is denoted by $B(B(\mathcal{H}))$.
A Lindbladian $\mathcal{L}$ is the generator of a quantum Markov semigroup $\{e^{t\mathcal{L}}\}_{t\geq 0}$. It takes a specific form
\[
\mathcal{L}[X] = \sum_{j=1}^r L_j^\dag X L_j - \frac{1}{2}\{L_j^\dag L_j,X\},
\]
where $\{\cdot\}$ denotes the anti-commutator, and each $L_j\in B(\mathcal{H})$ is called a \emph{jump operator}. 
We consider a local Lindbladian
\begin{equation}
    \label{eq:Lindbladian}
    \mathcal{L} = \sum_{m=1}^M \mathcal{L}_m,
\end{equation}
where each $\mathcal{L}_m$ is a Lindbladian that acts non-trivially only on $\mathsf{k}$ sites, and overlaps with at most $g$ other terms $\mathcal{L}_{m'}$. These Lindbladian terms do not have to commute with each other.
$\|\mathcal{L}_m\|_\diamond\leq 1$ where $\|\cdot\|_\diamond$ denotes the diamond norm.
We are working in the Heisenberg picture and therefore
\begin{equation}
\label{eq:trace_preserving}
    \mathcal{L}_1(I)=\cdots=\mathcal{L}_M(I)=\mathcal{L}(I)=0.
\end{equation}

We will define the $\sigma$-KMS inner product for a quantum state $\sigma\in B(\mathcal{H})$:
\begin{defn}[$\sigma$-KMS inner product]
\label{defn:KMS_inner_product}
    For any $X,Y\in B(\mathcal{H})$, we denote
    \[
    \braket{X,Y}_\sigma = \Tr[X^\dag \sqrt{\sigma} Y \sqrt{\sigma}].
    \]
    The corresponding norm we denote by
    \[
    \|X\|_\sigma = \braket{X,X}_\sigma^{1/2}.
    \]
\end{defn}
With this we can generalize the detailed balance condition (DBC) to the quantum setting
\begin{defn}[$\sigma$-KMS detailed balance]
\label{defn:KMS_detailed_balance}
    A superoperator $\mathcal{A}\in B(B(\mathcal{H}))$ satisfies $\sigma$-KMS detailed balance if it is self-adjoint under the $\sigma$-KMS inner product, i.e.,
    \[
    \braket{X,\mathcal{A}(Y)}_\sigma = \braket{\mathcal{A}(X),Y}_\sigma,
    \]
    for all $X,Y\in B(\mathcal{H})$.
\end{defn}

From now on we assume that each $\mathcal{L}_m$ satisfies $\sigma$-KMS DBC, which then implies
\[
\Tr[X\mathcal{L}_m^\dag(\sigma)] = \Tr[\mathcal{L}_m(X)\sigma] = \braket{\mathcal{L}_m(X),I}_\sigma = {\braket{X,\mathcal{L}_m(I)}_\sigma}=0,
\]
for all $X\in B(\mathcal{H})$.\footnote{Here we use $^\dag$ to denote conjugation under the Hilbert-Schmidt inner product.} This then implies $\mathcal{L}_m^\dag(\sigma)=0$ for all $m$. Therefore $\sigma$ is a stationary state for the dynamical semigroup generated by $\mathcal{L}^\dag$.

Because $\mathcal{L}$ is self-adjoint under the $\sigma$-KMS inner product, its spectrum must be real and located on the non-positive part of the real axis. $0$ must be an eigenvalue and here we assume it is non-degenerate. In particular, $\ker \mathcal{L} = \spanop{I}$. We denote the gap between $0$ and the second largest eigenvalue of $\mathcal{L}$ by $\mathrm{gap}(\mathcal{L})$. More generally
\begin{defn}[Spectral gap]
    For $\mathcal{A}\in B(B(\mathcal{H}))$ with real spectrum, we denote by $\mathrm{gap}(\mathcal{A})$ the gap between the two largest eigenvalues.
\end{defn}

The state $\sigma$ can then be prepared by evolving with $\mathcal{L}^\dag$ for sufficiently long time. However, this strategy does not necessarily translate to the most efficient quantum algorithm. Here we will consider another strategy. 
We approximately implement a quantum channel $\Phi^\dag$, where $\Phi$ is defined through
\begin{equation}
\label{eq:detectability_lem_op}
    \Phi = P_1 P_2\cdots P_M,\quad P_m = \lim_{t\to \infty} e^{t\mathcal{L}_m},\quad \forall 1\leq m\leq M.
\end{equation}
We will use the detectability lemma to show that $\Phi^\dag$ works almost as well as $e^{t\mathcal{L}^\dag}$, as a sufficient number of rounds of $\Phi^\dag$ brings any initial state close to the stationary state $\sigma$.
Note that in the above we are using an arbitrary ordering that does not depend on the underlying geometry of the quantum system.

\section{The detectability lemma for Lindbladians}
\label{sec:detectability_lemma}

The detectability lemma was originally formulated for the Hamiltonian situation. We will restate it here.
\begin{lem}[The detectability lemma \cite{AnshuAradVidick2016simple}]
    \label{lem:detectability_lemma}
    Let $Q_1,Q_2,\cdots,Q_M$ be a set of projection operators and $H=\sum_{m=1}^M Q_m$. We assume that each $Q_m$ commutes with all but at most $g$ others. Given a state $\ket{\psi}$, define $\ket{\phi}=\prod_{m=1}^M (I-Q_m)\ket{\psi}$, where the product is in any order. Let $\epsilon_\phi = \braket{\phi|H|\phi}/\braket{\phi|\phi}$, then
    \[
    \braket{\phi|\phi}\leq \frac{1}{\epsilon_\phi/g^2+1}.
    \]
\end{lem}
Note that the above lemma does not put any constraint on the locality on the operators $Q_m$, but only on their commutation relations. 
This lemma provides us a means to quantify how much the detectability lemma operator $\prod_{m=1}^M (I-Q_m)$ shrinks any quantum state that is orthogonal to the ground space of $H$, as stated in the following corollary
\begin{cor}[Corollary~3 of \cite{AnshuAradVidick2016simple}]
\label{cor:shrink}
    Under the same assumptions as Lemma~\ref{lem:detectability_lemma}, and assuming that the smallest positive eigenvalue of $H$ is at least $\gamma$, let $\ket{\psi^\perp}$ be any quantum state that is orthogonal to $\ker(H)$, which can be trivial, then
    \[
    \Big\|\prod_{m=1}^M (I-Q_m)\ket{\psi^\perp}\Big\|^2\leq \frac{1}{\gamma/g^2+1}.
    \]
\end{cor}
We remark that the above corollary is slightly different from \cite[Corollary~3]{AnshuAradVidick2016simple}, in that we did not assume the Hamiltonian $H$ to be frustration free, i.e., its kernel can be trivial. This is because the proof of \cite[Lemma~2]{AnshuAradVidick2016simple} and \cite[Corollary~3]{AnshuAradVidick2016simple} in fact did not use the frustration-free property of $H$.

Coming back to the Lindbladian setting, we can similarly construct a Hamiltonian superoperator
\begin{equation}
    \label{eq:ham_superpop}
    H_{\mathcal{L}} = \sum_{m=1}^M (I-P_m),
\end{equation}
for $P_m$ defined in \eqref{eq:detectability_lem_op}. Here $I$ denotes the identity superoperator. This Hamiltonian superoperator $H_{\mathcal{L}}$ is self-adjoint under the $\sigma$-KMS inner product and can therefore be regarded as a Hamiltonian. $H_{\mathcal{L}}$ has a non-degenerate ground space $\spanop(I)$: on the one hand $H_{\mathcal{L}}(I)=0$, and on the other hand if $H_{\mathcal{L}}(X)=0$, we then have $P_m(X)=X$, which implies $\mathcal{L}_m(X)=0$ for all $m$, thus forcing $X\propto I$.

Since $\mathcal{L}_m$ has all non-positive eigenvalues between $0$ and $1$, we have $I-P_m\succeq -\mathcal{L}_m$ (in the $\sigma$-KMS inner product sense), and consequently $H_{\mathcal{L}}\succeq -\mathcal{L}$. Additionally, because $H_{\mathcal{L}}$ and $-\mathcal{L}$ share the same ground space ($\mathrm{span}(I)$), denoting the gap between $0$ and the smallest positive eigenvalue of $H_{\mathcal{L}}$ by $\gamma$, we have
\begin{equation}
    \gamma \geq \mathrm{gap}(\mathcal{L}).
\end{equation}
With this we can directly apply Corollary~\ref{cor:shrink} to $H_{\mathcal{L}}$ to obtain
\begin{thm}
    \label{thm:quantum_gibbs_sampler_shrink}
    Let $\mathcal{L}=\sum_{m=1}^M \mathcal{L}_m$, where each $\mathcal{L}_m$ is a Lindbladian that satisfies $\sigma$-KMS DBC, and we assume $\mathcal{L}$ has a non-degenerate $0$-eigenspace. We assume that each $\mathcal{L}_m$ commutes with all but at most $g$ others. Then let the quantum channel $\Phi$ be defined as in \eqref{eq:detectability_lem_op}. For any $X\in B(\mathcal{H})$ such that $\Tr[\sigma X]=0$, we have
    \[
    \|\Phi(X)\|_\sigma^2\leq \frac{\|X\|_\sigma^2}{\mathrm{gap}(\mathcal{L})/g^2+1}.
    \]
    Moreover $\Tr[\sigma \Phi(X)]=\Tr[\Phi^\dag(\sigma) X]=\Tr[\sigma X]=0$.
\end{thm}

As a corollary of the above theorem, we can use the quantum channel $\Phi$ to prepare the Gibbs state $\sigma$ from any initial state $\rho_0$.
\begin{cor}
    \label{cor:prepare_gibbs_state_quantum_gibbs_sampler}
    Under the same assumptions as Theorem~\ref{thm:quantum_gibbs_sampler_shrink}, let $\rho_k=(\Phi^\dag)^k(\rho_0)$ for any initial state $\rho_0$, then
    \[
    \|\rho_k-\sigma\|_1\leq \frac{1}{(\mathrm{gap}(\mathcal{L})/g^2+1)^{k/2}\sqrt{\sigma_{\min}}},
    \]
    where $\sigma_{\min}$ is the smallest eigenvalue of $\sigma$.
\end{cor}

\begin{proof}
    Let $X\in B(\mathcal{H})$ satisfy $\|X\|\leq 1$. Denote $\Tilde{X} = X - \Tr[\sigma X] I$, and we have $\|\Tilde{X}\|_\sigma\leq \|X\|_\sigma\leq \|X\|\leq 1$. This ensures $\Tr[\sigma\tilde{X}]=0$. We then compute
    \[
    \begin{split}
    \Tr[X(\rho_k-\sigma)] = & \Tr[\Tilde{X}(\rho_k-\sigma)]=\Tr[\Tilde{X}\rho_k] \\ 
    & = \Tr[\Tilde{X}(\Phi^\dag)^k(\rho_0)] = \Tr[\Phi^k(\Tilde{X})\rho_0].
    \end{split}
    \]
    By the Cauchy-Schwarz inequality, we have
    \[
    \begin{split}
        \Tr[\Phi^k(\Tilde{X})\rho_0]| & =|\braket{\Phi^k(\Tilde{X}),\sigma^{-1/2}\rho_0\sigma^{-1/2}}_\sigma| \\ 
        & \leq \|\Phi^k(\Tilde{X})\|_\sigma\|\sigma^{-1/2}\rho_0\sigma^{-1/2}\|_\sigma.
    \end{split}
    \]
    For $\|\Phi^k(\Tilde{X})\|_\sigma$ we directly apply Theorem~\ref{thm:quantum_gibbs_sampler_shrink} $k$ times and use the fact that $\|\Tilde{X}\|_\sigma\leq 1$ and $\Tr(\sigma \Tilde{X}) = 0$. For $\|\sigma^{-1/2}\rho_0\sigma^{-1/2}\|_\sigma$ we have
    \[
    \|\sigma^{-1/2}\rho_0\sigma^{-1/2}\|_\sigma = \Tr[\rho_0\sigma^{-1/2}\rho_0\sigma^{-1/2}]^{1/2}\leq \frac{1}{\sqrt{\sigma_{\min}}}.
    \]
    Therefore we have
    \[
    |\Tr[X(\rho_k-\sigma)]|\leq \frac{1}{(\mathrm{gap}(\mathcal{L})/g^2+1)^{k/2}\sqrt{\sigma_{\min}}}.
    \]
    Because this is true for any $\|X\|\leq 1$, we have the desired inequality for the trace distance between $\rho_k$ and $\sigma$ through the duality between the Schatten $1$- and $\infty$-norms.
\end{proof}

We will next discuss the gate complexity of preparing Gibbs states using the above approach. We assume that each $\mathcal{L}_m$ acts non-trivially only on $\mathsf{k}=\Or(1)$ qubits. To prepare the Gibbs state $\sigma$ that is the fixed point of $\mathcal{L}$ as discussed above, we repeatedly apply the quantum channel $\Phi^\dag$. Each application of $\Phi^\dag$ requires implementing all $M$ channels $P_m$. Since each $P_m$ acts non-trivially only on $\mathsf{k}$ qubits, we can compile a unitary circuit involving at most $\mathsf{k}$ qubits to implement it. Solovay-Kitaev Theorem then tells us that $\Or(\log^c(1/\epsilon'))$ gates are needed to implement $P_m$ to accuracy $\epsilon'$ in the diamond norm.

Corollary~\ref{cor:prepare_gibbs_state_quantum_gibbs_sampler} tells us that in order to prepare the state to $\Or(\epsilon)$ accuracy in trace distance, it suffices to have $k=\Or(\frac{g^2}{\mathrm{gap}(\mathcal{L})}\log(1/(\sigma_{\min}\epsilon)))$. The accumulated unitary synthesis error in the whole procedure is at most $\Or(kM\epsilon')$, for which we can choose $kM\epsilon'=\Or(\epsilon)$ to ensure that it contributes only $\Or(\epsilon)$ to the total error. Therefore the gate complexity for the whole procedure is upper bounded by
\begin{equation}
    \label{eq:gate_complexity_cyclic}
    \begin{split}
    \Or\left(\frac{M g^2}{\mathrm{gap}(\mathcal{L})}\log(1/(\sigma_{\min}\epsilon))\log^c(kM/\epsilon)\right) \\ 
    =\wt{\Or}\left(\frac{M g^2}{\mathrm{gap}(\mathcal{L})}\log\left(\frac{1}{\sigma_{\min}\epsilon}\right)\log^c\left(\frac{M}{\epsilon}\right)\right)
    \end{split}
\end{equation}
We therefore arrive at the following corollary
\begin{cor}
\label{cor:gate_complexity}
    Let $\mathcal{L}=\sum_{m=1}^M \mathcal{L}_m$, where each $\mathcal{L}_m$ is a Lindbladian that satisfies $\sigma$-KMS DBC, and we assume $\mathcal{L}$ has a non-degenerate $0$-eigenspace. We assume that each $\mathcal{L}_m$ commutes with all but at most $g$ others. Then we can prepare $\sigma$ up to $\epsilon$ error in trace distance using $\wt{\Or}\left(\frac{M g^2}{\mathrm{gap}(\mathcal{L})}\log\left(\frac{1}{\sigma_{\min}\epsilon}\right)\log^c\left(\frac{M}{\epsilon}\right)\right)$ elementary single- and two-qubit gates, where $\sigma_{\min}$ is the minimum eigenvalue of $\sigma$, and $c$ is a constant from the Solovay-Kitaev theorem.
\end{cor}

Given the Lindbladian $\mathcal{L}$, the natural way to prepare $\sigma$ is to simply simulate the time evolution $e^{t\mathcal{L}}$ to bring an arbitrary initial state close to $\sigma$. The required evolution time is $\Or(\frac{1}{\mathrm{gap}(\mathcal{L})}\log(\frac{1}{\sigma_{\min}\epsilon}))$. Using state-of-the-art Lindbladian simulation algorithms such as \cite{DingLiLin2024simulating}, where the Lindbladian needs to be rescaled to have diamond norm $\Or(1)$, the gate complexity scales as $\Or(\frac{M^2}{\mathrm{gap}(\mathcal{L})}\log(\frac{1}{\sigma_{\min}\epsilon}))$ ignoring further poly-logarithmic factors. Therefore the method in Corollary~\ref{cor:gate_complexity} reduces the cost by a factor of $M$ up to poly-logarithmic corrections.

{This comparison with direct Lindbladian simulation uses the standard spectral-gap bound on the mixing time. If a stronger functional inequality or a model-specific argument proves a much shorter mixing time---for example, logarithmic mixing in the system size---then direct Lindbladian simulation can have a better guarantee than the gap-based bound above, and our present analysis does not imply an advantage over it. It remains an interesting open question whether such rapid-mixing estimates can be transferred directly to the discrete channel $\Phi^\dag$ \eqref{eq:detectability_lem_op} that we implement in our algorithm.}

\section{Implementing ground state projection operators}
\label{sec:implementing_ground_state_projection_operators}

In this section we will take a detour to discuss how one can implement the ground state projection operator of frustration-free Hamiltonians. In later sections we will discuss how this implementation leads to a quadratic speedup in terms of the spectral gap dependence in a simulated annealing procedure for quantum Gibbs states corresponding to commuting local Hamiltonians.

\begin{defn}
\label{defn:bounded_degree_local}
    A Hamiltonian $H = \sum_{m=1}^M H_m$
    acting on \(n\) qubits is called a \emph{bounded-degree $\mathsf{k}$-local Hamiltonian} if:
    (1) each term \(H_j\) acts non-trivially on at most $\mathsf{k}$ qubits, where \(\mathsf{k}=\mathcal{O}(1)\), and
    (2) there exists a constant \(g=\mathcal{O}(1)\) such that for every qubit \(i\in\{1,\dots,n\}\), at most \(g\) terms \(H_j\) act non-trivially on \(i\).
\end{defn}
{Here and throughout this section, $\mathsf{k}$ denotes the locality of a Hamiltonian term; the ordinary italic $k$ used in Section~\ref{sec:detectability_lemma} denotes the number of DL-channel iterations.}

We focus on such a bounded-degree local Hamiltonian on $n$ qubits 
\begin{equation}
    \label{eq:general_frsutration_free_ham}
    H = \sum_{m=1}^M H_m,
\end{equation}
where each $H_m$ satisfies $0\preceq H_m\preceq I$, is supported on at most $\mathsf{k}$ qubits, and overlaps with at most $g$ other terms $H_{m'}$. The terms $H_m$ do not necessarily commute with each other. Moreover, we assume that $H$ is frustration free in the following sense: let $P_H$ be the projection operator into the ground space of $H$, and then we have 
\begin{equation}
    \label{eq:frustration_free_ham}
    P_H H_m  = 0,\quad\forall m\in\{1,2,\cdots,M\}.
\end{equation}
Note that the above also implies $H_m P_H=0$, and that the ground space of $H$ is also its kernel. We also denote the dimension of the ground space by $\mathsf{r}$.

For the above Hamiltonian, we will implement its ground projection operator using the detectability lemma operator
\begin{equation}
    \label{eq:detectability_lemma_op_ham}
    \mathrm{DL}(H) = \prod_{m=1}^M P_m.
\end{equation}
Note that by applying Corollary~\ref{cor:shrink} to $\sum_m (I-P_m)$, we have, for any $\ket{\psi^\perp}$ such that $P_H \ket{\psi^\perp}=0$, 
\begin{equation}
    \label{eq:shrink_detectability_lemma_ham}
    \|\mathrm{DL}(H)\ket{\psi^\perp}\|^2\leq \frac{1}{\gamma/g^2+1},
\end{equation}
where $\gamma$ is the spectral gap separating $0$ from the rest of the spectrum of $H$. Note that the above detectability lemma operator $\mathrm{DL}(H)$ can be implemented with gate complexity $\Or(M)$ assuming that $\mathsf{k},g=\Or(1)$. Repeatedly applying $\mathrm{DL}(H)$ $\Or(\gamma^{-1}\log(1/\epsilon))$ times can take the resulting product $\epsilon$-close to the ground projection operator $P_H$, thus resulting in a $\Or(M\gamma^{-1}\log(1/\epsilon))$ gate complexity for implementing the projection operator. Below we will show how a $\gamma^{-1/2}$ scaling can be achieved using quantum singular value transformation (QSVT) \cite{Gilyen_2019singularvalue}.

To discuss QSVT we need to first introduce the notion of block encoding:
\begin{defn}[Block encoding]
\label{defn:block_encoding}
    An $(m+n)$-qubit unitary operator $U$ is called an $(\alpha, m, \epsilon)$-block-encoding of an $n$-qubit operator $A$, if 
    \begin{equation}
    \norm{A-\alpha(\bra{0^m}\otimes I) U (\ket{0^m}\otimes I)}\leq \epsilon.
    \label{eqn:block_encoding}
    \end{equation}
\end{defn}
An equivalent way to express \eqref{eqn:block_encoding} is 
\[
U=\begin{pmatrix}
\wt{A}/\alpha & * \\
* & *
\end{pmatrix},
\]
where $*$ can be any block matrices of the correct sizes and $\|\wt{A}-A\|\leq \epsilon$.

We first note that, because each $P_m$ can be implemented via its $(1,1,0)$-block encoding, these block encodings can be composed to implement a $(1,M,0)$-block encoding of $\mathrm{DL}(H)$. Using the compression gadget in \cite{LowWiebe2018hamiltonian,FangLinTong2023time} we can further reduce the number of ancilla qubit to $\Or(\log(M))$, thus resulting in a $(1,\Or(\log(M)),0)$-block encoding of $\mathrm{DL}(H)$. Note that each implementation of $P_m$ may involve exponentially small errors due to unitary synthesis, but these errors can be suppressed with polylogarithmic overhead, and their accumulation throughout the algorithm is well-controlled, as will be discussed in Corollary~\ref{cor:gate_complexity_approximate_proj}.

We then write down the singular value decomposition of $\mathrm{DL}(H)$:
\begin{equation}
    \mathrm{DL}(H) = U S V^\dag,
\end{equation}
where $S = \mathrm{diag}(s_1,s_2,\cdots,s_{2^n})$, $s_1\geq s_2\geq\cdots \geq s_{2^n}$.
Since the ground space of $H$ is $\mathsf{r}$-dimensional, we have $s_1=\cdots=s_{\mathsf{r}}=1$. By \eqref{eq:shrink_detectability_lemma_ham}, we have
\begin{equation}
    \label{eq:singular_val_bound}
    0\leq s_j\leq \frac{1}{\sqrt{\gamma/g^2+1}}, \quad j=\mathsf{r}+1,\cdots,2^n.
\end{equation}
We denote by $\gamma_*$ the lower bound of the gap separating the above $s_j$ from $1$:
\begin{equation}
    \gamma_* = 1-\frac{1}{\sqrt{\gamma/g^2+1}} = \Omega(\gamma).
\end{equation}

The projection operator $P_H$ can in fact be expressed in terms of $U$ and $V$. We write 
\begin{equation}
\label{eq:U1_V1}
    U = \begin{pmatrix}
    U_1 & U_2
    \end{pmatrix},
    \quad
    V = \begin{pmatrix}
        V_1 & V_2
    \end{pmatrix},
\end{equation}
where $U_1$ and $V_1$ are both of size $2^n\times \mathsf{r}$. In other words, the columns of $U_1$ and $V_1$ correspond to the singular values $s_1,\cdots,s_{\mathsf{r}}$. Then we claim
\begin{lem}
    \label{lem:projection_operator_in_terms_of_UV}
    For $U_1$ and $V_1$ defined above, we have $P_H= U_1 V_1^\dag$.
\end{lem}
\begin{proof}
    For any $\ket{\phi}\in \mathrm{col}(V_1)$, we have $\|\mathrm{DL}(H)\ket{\phi}\|=\|\ket{\phi}\|$, and therefore $\|P_m\ket{\phi}\|=\|\ket{\phi}\|$ for all $m$. Since each $P_m$ is a projection operator, we then have $P_m\ket{\phi}=\ket{\phi}$. Consequently
    \[
    U_1 V_1^\dag\ket{\phi}=\mathrm{DL}(H)\ket{\phi}=\ket{\phi},
    \]
    and $\ket{\phi}$ is in the ground space of $H$. As a result of the latter $P_H\ket{\phi}=\ket{\phi}$. Therefore $P_H$ and $U_1V_1^\dag$ agree in $\mathrm{col}(V_1)$.
    
    If $\ket{\phi}\in \mathrm{col}(V_1)^\perp$, then $\ket{\phi}$ must be orthogonal to the ground space since the ground space is contained in $\mathrm{col}(V_1)$. Therefore $P_H\ket{\phi}=U_1V_1^\dag\ket{\phi}=0$. Combining this with the discussion in the previous paragraph we have $P_H=U_1 V_1^\dag$.
\end{proof}

For a polynomial $p(x)$ of degree $d$ with fixed parity satisfying $|p(x)|\leq 1$ for all $x\in[-1,1]$, QSVT allows us to construct a $(1,\Or(\log(M)),0)$-block encoding of 
\begin{equation}
    \tilde{P}_H = U p(S) V^\dag,
\end{equation}
with $d$ queries to the block encoding of $\mathrm{DL}(H)$. We will appropriately choose $p(x)$ so that $\tilde{P}_H$ approximates $P_H$. For this purpose, the polynomial in \cite{ThibodeauClark2023nearly} suffices, but we will provide a simpler construction by simply rescaling the Chebyshev polynomial. We let 
\begin{equation}
    \label{eq:choice_for_poly_p}
    p(x) = \frac{T_{\ell}\left(x/(1-\gamma_*)\right)}{T_{\ell}\left(1/(1-\gamma_*)\right)},
\end{equation}
where $T_\ell$ is the $\ell$th Chebyshev polynomial of the first kind, and is of degree $\ell$ with fixed parity. It satisfies $|p(x)|\leq 1$ for $x\in[-1,1]$, $p(1)=1$, and when $x\in[-1+\gamma_*,1-\gamma_*]$, we have
\[
|p(x)|\leq \frac{1}{T_{\ell}\left(1/(1-\gamma_*)\right)}\leq 2e^{-\ell \sqrt{1/(1-\gamma_*)-1}}\leq 2e^{-\ell\sqrt{\gamma_*}},
\]
where the second inequality comes from \cite[Lemma~13]{LinTong2020optimal}. From this we have
\[
\|p(S)-D\|\leq 2e^{-\ell\sqrt{\gamma_*}},
\]
where $D=\mathrm{diag}(1,\cdots,1,0,\cdots,0)$, with $\mathsf{r}$ diagonal entries being $1$. The above implies
\begin{equation}
    \label{eq:approximate_proj_op_error_bound}
    \begin{split}
    \|\tilde{P}_H-P_H\|=\|\tilde{P}_H-U_1 V_1^\dag\|=\|\tilde{P}_H-UDV^\dag\|
    \\ =\|p(S)-D\|\leq 2e^{-\ell\sqrt{\gamma_*}}.
    \end{split}
\end{equation}
In order to ensure $\|\tilde{P}_H-P_H\|\leq \epsilon$, we only need
\begin{equation}
    \label{eq:choice_of_ell}
    \ell = \Or(\gamma_*^{-1/2}\log(1/\epsilon))=\Or(\gamma^{-1/2}\log(1/\epsilon)).
\end{equation}
Therefore we have the following theorem:
\begin{thm}
    \label{thm:block_encode_proj}
    Let $H = \sum_{m=1}^M H_m$, where each $H_m$ satisfies $0\preceq H_m\preceq I$, is supported on at most $\mathsf{k}$ qubits, and overlaps with at most $g$ other terms $H_{m'}$. We assume that $H$ is frustration-free in the sense that $P_H H_m=0$ for all $m=1,2,\cdots,M$, where $P_H$ is its ground space projection operator. 
    We assume that each ground projection operator of $H_m$ is accessed through its $(1,\Or(1),0)$-block encoding.
    Then a $(1,\Or(\log(M)),\epsilon)$-block encoding of $P_H$ can be obtained using $\Or(\gamma^{-1/2}\log(1/\epsilon))$ queries to each block encoding of the ground space projection operator of $H_m$.
\end{thm}
We note that each block encoding of the ground space projection operator of $H_m$ needs to be implemented with gates and has some errors. Suppose each introduces an $\epsilon'$ error in terms of the spectral norm, then the cumulative error in the block encoding of $P_H$ is $\Or(M\ell \epsilon')$. To ensure that the error in terms of the spectral norm is $\epsilon'$ for each small block encoding involving $\Or(1)$ qubits, the Solovay-Kitaev theorem tells us that $\Or(\log^c(1/\epsilon'))$ gates are needed, where $c=\log_{(1+\sqrt{5})/2}(2)+\delta$ for any $\delta>0$. Therefore
\begin{cor}
    \label{cor:gate_complexity_approximate_proj}
    Let $H = \sum_{m=1}^M H_m$, where each $H_m$ satisfies $0\preceq H_m\preceq I$, is supported on at most $\mathsf{k}$ qubits, and overlaps with at most $g$ other terms $H_{m'}$. We assume that $H$ is frustration-free in the sense that $P_H H_m=0$ for all $m=1,2,\cdots,M$, where $P_H$ is its ground space projection operator. Then a $(1,\Or(\log(M)),\epsilon)$-block encoding of $P_H$ can be obtained using 
    \[
    \Or\left(M\gamma^{-1/2}\log(1/\epsilon)\log^c(M\gamma^{-1/2}\log(1/\epsilon)/\epsilon)\right)
    \]
    elementary gates, where $c$ is the exponent in the Solovay-Kitaev theorem.
\end{cor}

\section{The parent Hamiltonian}
\label{sec:parent_hamiltonian}

The parent Hamiltonian is a coherent representation of the Lindbladian $\mathcal{L}$ which allows for a quadratic speedup in terms of the spectral gap dependence at the expense of doubling the Hilbert space. It is constructed such that its ground state provides access to the Gibbs state. The Hermiticity of $\mathcal{L}$ under the $\sigma$-KMS inner product results in the Hermiticity of this coherent representation, which is therefore viewed as a Hamiltonian. This representation enables us to prepare the Gibbs state by preparing the ground state of the parent Hamiltonian using the techniques in Section~\ref{sec:implementing_ground_state_projection_operators} and obtain a quadratic speedup.

Define the superoperator $\Gamma_\sigma: \calB(\calH) \to \calB(\calH)$ as
\begin{equation}
    \label{eq:gamma_sigma_def}
    \Gamma_\sigma (\cdot) = \sigma^{\frac{1}{2}} \cdot \sigma^{\frac{1}{2}}.
\end{equation}
$\Gamma_\sigma$ recovers the $\sigma$-KMS inner product in the sense that
\begin{equation}
    \braket{X,Y}_{\mathrm{KMS}} = \braket{\Gamma_\sigma^\frac{1}{2} X, \Gamma_\sigma^\frac{1}{2} Y}_{\mathrm{HS}}.
\end{equation}
Define the superoperator $\frakH: \calB(\calH) \to \calB(\calH)$ as 
\begin{equation}
    \frakH = \Gamma_\sigma^\frac{1}{2} \calL \Gamma_\sigma^{-\frac{1}{2}}.
\end{equation}
This superoperator is Hermitian under the usual Hilbert-Schmidt inner product, i.e., 
\begin{equation}
     \braket{\frakH X, Y}_{\mathrm{HS}} = \braket{X, \frakH Y}_{\mathrm{HS}}, \;\;\; \forall X,Y \in \calB(\calH).
\end{equation}
By vectorizing the operator space $\calB(\calH)$, the superoperator $\frakH$ can be explicitly written as a Hamiltonian. Define the linear map $v : \calB(\calH) \to \calH^{\otimes 2}$ that acts on the computational basis as follows:
\begin{equation}
    v (\ket{i} \bra{j}) = \ket{i} \otimes \ket{j}, \;\;\; i,j \in [d].
\end{equation}
By definition and linearity, we have
\begin{equation}
    v(X) = v(\sum_{i,j} X_{ij} \ket{i} \bra{j}) = \sum_{i,j} X_{ij} \ket{i} \otimes \ket{j}, \;\;\; X \in \calB(\calH).
\end{equation}
In particular, for a rank-$1$ operator $X = \ket{\psi} \bra{\phi}$, we have
\begin{equation}
    v(\ket{\psi} \bra{\phi}) = \ket{\psi} \otimes \ket{\phi^*},
\end{equation} where $\ket{\phi^*}$ is the entrywise complex conjugate of $\ket{\phi}$ in the computational basis.
{Indeed, writing $\ket{\psi}=\sum_i\psi_i\ket{i}$ and $\ket{\phi}=\sum_j\phi_j\ket{j}$ gives $X_{ij}=\psi_i\overline{\phi_j}$, so the preceding definition yields $v(X)=\sum_{i,j}\psi_i\overline{\phi_j}\ket{i}\otimes\ket{j}=\ket{\psi}\otimes\ket{\phi^*}$.} 
Upon vectorization, the superoperator $\frakH$ is represented as a Hamiltonian $\sfH$ acting on the doubled Hilbert space $\calH^{\otimes 2}$,
\begin{equation}\label{eq:superoperator_vectorization}
    \sfH = v \, \frakH \, v^{-1}.
\end{equation}
Observe that
\begin{equation}
    \frakH(\sqrt{\sigma}) = \Gamma_\sigma^{1/2} \calL (I) = 0.
\end{equation} Therefore $\sqrt{\sigma}$ is the ground\footnote{Strictly speaking, this state is the top eigenstate of the parent Hamiltonian, but they are equivalent up to flipping sign.} state of the superoperator $\frakH$ with eigenvalue $0$. Or equivalently, the vectorized state $v(\sqrt{\sigma})$ is the ground state of the Hamiltonian $\sfH$ with eigenvalue $0$.
Moreover, assuming that the eigenvalue decomposition of $\sigma$ is given by $\sigma = \sum_{i=1}^d \lambda_i \ket{\psi_i} \bra{\psi_i}$, 
\begin{equation}
    v(\sqrt{\sigma}) = \sum_{i=1}^d \sqrt{\lambda_i} \ket{\psi_i} \otimes \ket{\psi_i^*},
\end{equation} which is a purification of the Gibbs state $\sigma$, i.e., tracing over the second subsystem gives $\sigma$ on the first subsystem. Since the map from the Lindbladian to the Hamiltonian is linear, given a Lindbladian that is a sum of $M$ local Lindbladians, $\sfH$ is equal to the sum of $M$ local terms, each of which can be classically computed in time that is upper bounded by some constant depending on locality. The problem of preparing the Gibbs state $\sigma$ now reduces to preparing the ground state of the Hamiltonian $\sfH$.

\section{State preparation by annealing}
\label{sec:annealing}

To prepare the Gibbs state, we utilize a discretized temperature path similar to quantum annealing, and prepare the Gibbs state by projecting the state consecutively to the ground space of the parent Hamiltonian of the next temperature. Denote the purified Gibbs state at inverse temperature $\nu$ by $\ket{\psi_{\nu}}$. Denote by $P_{\nu} := \ket{\psi_\nu}\bra{\psi_\nu}$ the exact projection operator and $\Tilde{P}_{\nu}$ the approximate projection operator. We first prepare the ground state $\ket{\psi_{\beta_0}}$ of the parent Hamiltonian at $\beta_0 = 0$, which is the purified maximally mixed state, or in other words the maximally entangled state on $\calH^{\otimes 2}$.
In the annealing process, we go through $K$ temperature steps,
\begin{equation}
    \beta_0 = 0 \rightarrow \beta_1 \rightarrow \beta_2 \rightarrow \cdots \rightarrow \beta_K = \beta, \;\;\;\; \beta_{j-1} < \beta_{j} \quad \forall j\in[K].
\end{equation} 
Note that for each temperature, the associated parent Hamiltonian $\sfH_{\beta_j}$ is different, while the original Hamiltonian for the actual physical system as in the Gibbs state does not depend on temperature. 
For each temperature, using Corollary~\ref{cor:gate_complexity_approximate_proj}, we construct each approximate projection operator $\Tilde{P}_{\beta_j}$. 
Then, we send the initial state to the final state using $K$ transition operators, each constructed from two consecutive projectors,
\begin{equation}
    \label{eq:transitionops}
    O_{\beta_0}^{\beta_{1}} \rightarrow O_{\beta_1}^{\beta_2} \rightarrow \cdots \rightarrow O_{\beta_{K-1}}^{\beta_{K}}.
\end{equation} 
Denote the ideal transition operator by $O_{\beta_{j-1}}^{\beta_{j}} := \ket{\psi_{\beta_{j}}} \bra{\psi_{\beta_{j-1}}}$.
Error is introduced in the implementation of each $\Tilde{P}_{\beta_j}$ as well as in the implementation of the transition operator. Denote the transition operator with error by $\Tilde{O}_{\beta_{j-1}}^{\beta_{j}}$.

Denote the final state we obtain by $\ket{\psi}$,
\begin{equation}
    \label{eq:finalstate}
    \ket{\psi} = \frac{\prod_{j=1}^K \Tilde{O}_{\beta_{j-1}}^{\beta_{j}} \ket{\psi_{\beta_0}}}{\left\|\prod_{j=1}^K \Tilde{O}_{\beta_{j-1}}^{\beta_{j}} \ket{\psi_{\beta_0}} \right\|}.
\end{equation} 
$\ket{\psi}$ is obtained with success probability
\begin{equation}
    \label{eq:successprob}
    p = \left\|\prod_{j=1}^K \Tilde{O}_{\beta_{j-1}}^{\beta_{j}} \ket{\psi_{\beta_0}}\right\|^2.
\end{equation}
For convenience, denote the unnormalized state by
\begin{equation}
    \ket{\Tilde{\psi}_{\beta_i}} =  \prod_{j=1}^{i} \Tilde{O}_{\beta_{j-1}}^{\beta_{j}} \ket{\psi_{\beta_0}}.
\end{equation}
Let the error of each $\Tilde{O}_{\beta_{j-1}}^{\beta_{j}}$ be upper bounded by $\delta/(2K)$. i.e., $\|\Tilde{O}_{\beta_{j-1}}^{\beta_{j}} - O_{\beta_{j-1}}^{\beta_{j}}\| \le \delta/{2K}$. 
The error of the final state is then bounded by $\delta + \calO(\delta^2)$ due to triangle inequality. More precisely,
\begin{equation}
    \begin{aligned}
    & \| \ket{\Tilde{\psi}_\beta} - \ket{\psi_\beta} \| \\ 
    = & \norm{ \Tilde{O}_{\beta_{K-1}}^{\beta_{K}} (\ket{\Tilde{\psi}_{\beta_{K-1}}} - \ket{\psi_{\beta_{K-1}}}) + (\Tilde{O}_{\beta_{K-1}}^{\beta_{K}} - {O}_{\beta_{K-1}}^{\beta_{K}}) \ket{\psi_{\beta_{K-1}}} } \\
    \le & \norm{\ket{\Tilde{\psi}_{\beta_{K-1}}} - \ket{\psi_{\beta_{K-1}}}} + \norm{(\Tilde{O}_{\beta_{K-1}}^{\beta_{K}} - {O}_{\beta_{K-1}}^{\beta_{K}}) \ket{\psi_{\beta_{K-1}}}} \\ 
    \le & \cdots \le \sum_{j=1}^{K} \| \Tilde{O}_{\beta_{j-1}}^{\beta_{j}} - {O}_{\beta_{j-1}}^{\beta_{j}} \| \le \frac{\delta}{2},
    \end{aligned}
\end{equation} where from the first to the second line we additionally used the fact that $\|\Tilde{O}_{\beta_{j-1}}^{\beta_{j}}\|\le 1$ since $\Tilde{O}_{\beta_{j-1}}^{\beta_{j}}$ is encoded in a unitary. It immediately follows that $\| \ket{\Tilde{\psi}_\beta} \| \ge 1 - \frac{\delta}{2}$. 
The error of the final state
\begin{equation}
    \begin{aligned}
        \| \ket{\psi} - \ket{\psi_\beta}\| & = \left\| \frac{   \ket{\Tilde{\psi}_\beta}  }{  \|\ket{\Tilde{\psi}_\beta}\|  } - \ket{\psi_\beta}  \right\| \\ 
        & = \left\| \frac{ \ket{\Tilde{\psi}_\beta} - \|\ket{\Tilde{\psi}_\beta}\| \ket{\psi_{\beta}}  }{  \|\ket{\Tilde{\psi}_\beta}\|  } \right\| \\ 
        & = \frac{ \left\| \ket{\psi_\beta} ( 1 - \| \ket{\Tilde{\psi}_\beta} \|) + ( \ket{\Tilde{\psi}_\beta} - \ket{\psi_\beta} )\right\| }{  \|\ket{\Tilde{\psi}_\beta}\|  } \\ 
        & \le \delta + \calO(\delta^2).
    \end{aligned}
\end{equation}
On the other hand, the success probability
\begin{equation}
    p = \| \ket{\Tilde{\psi}_\beta} \|^2  \ge (1-\frac{\delta}{2})^2 \ge 1 - \delta.
\end{equation}
We summarize the above in the following lemma.
\begin{lem}
    \label{lem:finalerror}
    Let $\{ \Tilde{O}_{\beta_{j-1}}^{\beta_{j}} \}_{j=1}^K$ be a set of transition operators as in Eq.~\eqref{eq:transitionops}, the final state and the success probability defined in Eq.~\eqref{eq:finalstate} and Eq.~\eqref{eq:successprob}. Then, if for each transition operator $\|\Tilde{O}_{\beta_{j-1}}^{\beta_{j}} - \ket{\psi_{\beta_{j}}} \bra{\psi_{\beta_{j-1}}}\| \le \delta/(2K)$, the final state $\|\ket{\psi} - \ket{\psi_\beta}\| \le \delta + \calO(\delta^2)$ and the success probability $p \ge 1 - \delta$.
\end{lem}

Hence it suffices to construct $ \Tilde{O}_{\beta_{j-1}}^{\beta_{j}} $ to precision $\delta/(2K)$. We use QSVT to construct  $ \Tilde{O}_{\beta_{j-1}}^{\beta_{j}} $ from $\Tilde{P}_{\beta_j} \Tilde{P}_{\beta_{j-1}} \approx {\ket{\psi_{\beta_{j}}} \braket{ \psi_{\beta_j} | \psi_{\beta_{j-1}}  } \bra{\psi_{\beta_{j-1}}}}$,
by using it to boost
$|\braket{ \psi_{\beta_j} | \psi_{\beta_{j-1}}}|$ close to $1$.
Query complexity depends primarily on the overlap $|\braket{ \psi_{\beta_j} | \psi_{\beta_{j-1}}}|$, which is to be controlled by choosing an appropriate annealing path. The following lemma gives the query complexity of implementing  $ \Tilde{O}_{\beta_{j-1}}^{\beta_{j}} $ with $\Tilde{P}_{\beta_j}$, $\Tilde{P}_{\beta_{j-1}}$.

\begin{lem}
    \label{lem:transop}
    Given {the block-encodings of} approximate projectors $\tilde{P}_{\beta_j}$ and $\tilde{P}_{\beta_{j-1}}$ where 
    {$\|\tilde{P}_{\beta_j} - \ket{\psi_{\beta_j}} \bra{\psi_{\beta_j}}\|$ and $\|\tilde{P}_{\beta_{j-1}} - \ket{\psi_{\beta_{j-1}}} \bra{\psi_{\beta_{j-1}}}\|$ are both upper bounded by $0\le \mu < 1$}.
    Assume there exists a constant $b>0$ that lower bounds the overlap between these two states, $| \langle \psi_{\beta_{j}} | \psi_{\beta_{j-1}} \rangle | \ge b$. 
    Then, 
    using {one additional} ancilla qubit and $l= \calO (b^{-1} \log (\epsilon^{-1}))$ queries to {the block-encodings of} $\tilde{P}_{\beta_j}$ and $\tilde{P}_{\beta_{j-1}}$, 
    {we can construct the block-encoding of} $\Tilde{O}_{\beta_{j-1}}^{\beta_{j}}$ such that
    {$\|\Tilde{O}_{\beta_{j-1}}^{\beta_{j}} - \ket{\psi_{\beta_{j}}} \bra{\psi_{\beta_{j-1}}}\| \le 4l\sqrt{3\mu} + \epsilon$}.
\end{lem}
\begin{proof}
    
Let $U$ be the block-encoding of $P_{\beta_{j}}P_{\beta_{j-1}} = \ket{\psi_{\beta_{j}}} \braket{ \psi_{\beta_j} | \psi_{\beta_{j-1}}  } \bra{\psi_{\beta_{j-1}}}$ such that $\bra{\boldsymbol{0}} U \ket{\boldsymbol{0}} = P_{\beta_{j}}P_{\beta_{j-1}}$, where $\ket{\boldsymbol{0}}$ is the all-zero ancilla state.
$U$ can be obtained with one use of each of the block-encodings of $P_{\beta_{j}}$ and $P_{\beta_{j-1}}$. Let $\Pi = \ket{\boldsymbol{0}}\bra{\boldsymbol{0}} \otimes I$, where $I$ is the identity operator on the system. $\Pi U \Pi = W\Sigma V^\dag$ has a single non-zero singular value $ | \braket{ \psi_{\beta_j} | \psi_{\beta_{j-1}}  } | \ge b$. Following \cite{Gilyen_2019singularvalue}, we perform QSVT with an odd polynomial $p(\cdot)$ of degree $l = \mathcal{O} ( b^{-1}\log(1/\epsilon))$ such that $p$ is $\epsilon$-close to $1$ on the interval $[b, 1]$. This lets us construct a circuit $U_\Phi$ with $l$ uses of $U, U^\dagger$, such that $\Pi U_\Phi \Pi = p^{\text{SV}}(\Pi U \Pi) = W p(\Sigma) V^\dag$. Because of our choice of $p$, we have $ \| \Pi U_\Phi \Pi - \ket{\boldsymbol{0}}\bra{\boldsymbol{0}} \otimes \ket{\psi_{\beta_{j}}} \bra{\psi_{\beta_{j-1}}} \| \le \epsilon$.

Now we are given block encodings of approximate projectors $\Tilde{P}_{\beta_j}, \Tilde{P}_{\beta_{j-1}}$ rather than $P_{\beta_j},P_{\beta_{j-1}}$. Let $\widetilde{U}$ be the block encoding of $\Tilde{P}_{\beta_j}\Tilde{P}_{\beta_{j-1}}$.
The robustness of QSVT \cite[Lemma 22]{Gilyen_2019singularvalue} shows that for matrices $A$ and $\Tilde{A}$ of operator norm at most $1$, it holds that $\|p^{\text{SV}}(A) - p^{\text{SV}}(\Tilde{A})\| \le 4l \sqrt{\|A - \Tilde{A}\|}$. 
    Since
    \begin{equation}
        {\begin{aligned}
        & \| \tilde{P}_{\beta_j}\tilde{P}_{\beta_{j-1}} - P_{\beta_j} P_{\beta_{j-1}} \| \\
        = \, & \| (P_{\beta_j} + \delta P_{\beta_j}) (P_{\beta_{j-1}} + \delta P_{\beta_{j-1}}) - P_{\beta_j} P_{\beta_{j-1}}\| \\
        = \, & \| P_{\beta_j} \delta P_{\beta_{j-1}} + \delta P_{\beta_j} P_{\beta_{j-1}} + \delta P_{\beta_j} \delta P_{\beta_{j-1}} \| \\
        \le \,& 2\mu + \mu^2 \le 3\mu,
        \end{aligned}}
    \end{equation}
    where $\delta P_{\beta_j}=\Tilde{P}_{\beta_j}-P_{\beta_j}$ and $\delta P_{\beta_{j-1}}=\Tilde{P}_{\beta_{j-1}}-P_{\beta_{j-1}}$,
    we have $\|p^{\text{SV}}(\Pi U \Pi) - p^{\text{SV}}(\Pi \widetilde{U} \Pi)\| \le 4l \sqrt{3\mu}$. Finally, by the triangle inequality,
    \[
    \|p^{\text{SV}}(\Pi \widetilde{U} \Pi) - \ket{\boldsymbol{0}}\bra{\boldsymbol{0}} \otimes \ket{\psi_{\beta_j}} \bra{\psi_{\beta_{j-1}}}\| \le 4l\sqrt{3\mu} + \epsilon.
    \]
    Let $\widetilde{U}_\Phi$ be the circuit that implements $p^{\text{SV}}(\Pi \widetilde{U} \Pi)$, i.e., $\Pi \widetilde{U}_\Phi \Pi = p^{\text{SV}}(\Pi \widetilde{U} \Pi)$. Thus $\widetilde{U}_\Phi$ is the block encoding of $\Tilde{O}_{\beta_{j-1}}^{\beta_j}$ and uses $\widetilde{U},\widetilde{U}^\dag$ a total of $\mathcal{O}(b^{-1}\log(1/\epsilon))$ times.
\end{proof}

By choosing $\epsilon \le \delta/(4K)$ and $\sqrt{\mu} \le \delta/(16\sqrt{3} l K)$, we have {$\|\Tilde{O}_{\beta_{j-1}}^{\beta_j}-\ket{\psi_{\beta_j}}\bra{\psi_{\beta_{j-1}}}\|\leq\delta/(2K)$.} Also, when $\epsilon \le \delta/(4K)$, $l = \calO(b^{-1} \log(K/\delta))$. Substituting into $\mu$, we have $\mu = \calO\left( \left( \frac{b}{\frac{K}{\delta} \log \frac{K }{\delta}} \right)^2 \right)$.
To control $b$, we adopt the following proposition from \cite{chen2023quantumthermalstatepreparation} which gives a lower bound on the overlap between two consecutive Gibbs states.
\begin{lem}[Proposition G.2 of \cite{chen2023quantumthermalstatepreparation}]
    \label{lem:overlap}
    Let $\ket{\psi_{\beta}}$ and $\ket{\psi_{\beta+\delta\beta}}$ be the ground states of the parent Hamiltonians $\sfH_\beta$ and $\sfH_{\beta+\delta\beta}$ respectively, or in other words, the purified Gibbs states of Hamiltonian $H$ at temperature $\beta$ and $\beta + \delta\beta$. Then
    \begin{equation}
        |\braket{\psi_{\beta+\delta\beta} | \psi_{\beta}}|^2 = 1 - \calO( {\delta\beta}^2 \norm{H}^2  ).
    \end{equation}
\end{lem}
By uniformly spacing the temperatures choosing $K = \Theta(\beta \norm{H})$, we can make $b$ a constant. More precisely, choosing $K = \alpha \beta \norm{H}$ for some constant $\alpha > 1$ gives $b^2 = 1 - \calO(1/\alpha^2)$, which is a $\Theta(1)$ constant. 

In total, we need
\begin{equation}
    \calO \left(  \frac{M}{\sqrt{\gamma}} \log \frac{1}{\mu} \log^{c} \left(   \frac{M}{\sqrt{\gamma} \mu} \log \frac{1}{\mu}   \right)  \right) \cdot
    \calO \left( \frac{1}{b} \log \frac{K}{\delta} \right) \cdot K
\end{equation} uses of elementary gates, where $c$ is the exponent in the Solovay-Kitaev theorem, and $\gamma = \min_j \text{gap}(\sfH_{\beta_j})$. The first term comes from Corollary~\ref{cor:gate_complexity_approximate_proj}, which is the cost of constructing each approximate projector $\tilde{P}_{\beta_j}$ with error $\mu$. The second term comes from Lemma~\ref{lem:transop}, which is the cost of constructing each transition operator $O_{\beta_{j-1}}^{\beta_j}$ with error $\delta/(2K)$ given access to $\tilde{P}_{\beta_j}$ and $\tilde{P}_{\beta_{j-1}}$. The last term $K$ is the number of steps.
Substituting corresponding values, we have the following theorem.
\begin{thm}
    \label{thm:gibbsstateprepcomplex}
    For a given Hamiltonian $H = \sum_{m=1}^M H_m$, we assume: 
    \begin{enumerate}
        \item There exists Lindbladians $\mathcal{L}_\nu$, for $\nu\in[0,\beta]$, where each $\mathcal{L}_\nu$ is irreducible and satisfies the $\sigma_\nu$-detailed balance condition. Here $\sigma_\nu\propto e^{-\nu H}$ is the Gibbs state of $H$ at inverse temperature $\nu$.
        \item The parent Hamiltonian $\mathsf{H}_\nu$ of each $\mathcal{L}_\nu$ is frustration-free and bounded-degree local (Definition~\ref{defn:bounded_degree_local}).
    \end{enumerate}
    Then there exists a quantum algorithm that prepares the purified Gibbs state at inverse temperature $\beta$ with error $\delta$ and success probability at least $1-\delta$ using
    \begin{equation}
        \label{eq:gibbsstateprepruntime}
        \calO \left( \frac{M\beta \norm{H}}{\sqrt{\gamma}} \log^2 \left(\frac{\beta\norm{H}}{\delta}\right) \log^c \left( \frac{M}{\sqrt{\gamma}} \frac{\beta\norm{H}}{\delta} 
        \right) \right)
    \end{equation} elementary gates, and 
    \begin{equation}
        \label{eq:gibbsstateprepancillacnt}
        \calO (\log M) + 1
    \end{equation} resettable ancilla qubits,
    where $c$ is the exponent in the Solovay-Kitaev theorem, and $\gamma = \min_j \mathrm{gap}(\mathcal{L}_{\beta_j})$, $\beta_j = j\beta/K$, $K=\Theta(\beta\|H\|)$.
\end{thm}
The $\calO (\log M)$ ancilla qubits above are used to coherently track the success of each projection operator involved in the DL operator (see the discussion under Definition~\ref{defn:block_encoding} of block-encoding).

\section{Application to commuting bounded-degree local Hamiltonians}
\label{sec:commuting_ham}

{Commuting Hamiltonians include classical spin models, such as Ising- and Potts-type systems, as well as commuting-projector and stabilizer Hamiltonians, including topologically ordered models such as the toric code. Preparing their Gibbs states provides coherent access to finite-temperature ensembles for estimating thermal observables and correlation functions, probing thermal phase structure and the thermal stability or loss of topological order, and serving as a subroutine in annealing-based estimation of partition-function ratios or free-energy differences. This class is also a standard test bed for studying thermalization and the stability of topological quantum memories against thermal noise \cite{KastoryanoBrandao2016quantum,HwangJiang2024gibbs,ding2024polynomial,alicki2009thermalization,temme2015fast,LuHsiehGrover2020}.}

In this section,
we consider the KMS-DBC Lindbladian of the form in \cite{Ding2025EfficientQuantumGibbs}, which encompasses the Davies generator \cite{Davies74,Davies76} and the algorithmic Lindbladian studied in \cite{ChenKastoryanoGilyen2023efficient}. The construction in \cite{Ding2025EfficientQuantumGibbs} is not, in general, always irreducible. Nevertheless, we shall assume irreducibility, as this is the case for the Davies generator and the algorithmic Lindbladian in \cite{ChenKastoryanoGilyen2023efficient}.  We will show that for commuting bounded-degree local Hamiltonians, the parent Hamiltonian is bounded-degree local and frustration free, therefore our method applies. 

To begin with, we review the Lindbladian form in \cite{Ding2025EfficientQuantumGibbs}.
A jump operator $L^a$ is assigned to each coupling operator $A^a$, defined as 
\begin{equation}
    L^a = \intinf f^a(t) A^a(t) \dd t , \quad A^a(t) = e^{iHt} A^a e^{-iHt},
\end{equation} with a weighting function $f^a: \bbR \to \bbC$ given by 
\begin{equation}
    f^a(t) = \frac{1}{2\pi } \intinf q^a(\omega) e^{-\beta \omega} e^{-it\omega} \dd \omega.
\end{equation} 
The function $q^a(\omega): \bbR \to \bbC$ satisfies $q^a(\omega) = \overline{q^a(-\omega)}$, and $\beta$ is the inverse temperature. 
$A^a(t)$ is the time evolution of $A^a$ under $H$. Note that the weighting function and hence the jump operators are in general temperature-dependent. For simplicity, we omit the subscript $\beta$ representing the temperature for $f^a$ and $L^a$ for now. 
In addition to the jump operators, there is a coherent term
\begin{equation}
    \label{eq:coherent_term_expression}
    G = \sum_{a\in \calA } \intinf g(t) \: e^{iHt} \, ((L^a)^\dag L^a) \, e^{-iHt} \dd t,
\end{equation} with $g(t):\bbR\to \bbC$ defined as
\begin{equation}
    g(t) = \frac{1}{2\pi } \intinf -\frac{i}{2} \tanh \left(-\frac{ \omega}{4} \right) \kappa(\omega) e^{-it\omega} \dd \omega,
\end{equation}  
where $\kappa(\omega)$ satisfies
\begin{equation}
    \kappa(\omega) = 1 \text{  for  } \omega \in [-2\norm{H}, 2 \norm{H}].
\end{equation}
Function $f^a$ and $g$ are all $L^1$-integrable. Finally, the total Lindbladian in the Heisenberg picture is defined as 
\begin{equation}
    \calL(X) = i[G,X] + \sum_{a\in\calA} \Big ( (L^a)^\dag X L^a - \frac{1}{2} \left\{ (L^a)^\dag L^a, X \right\} \Big).
\end{equation}

In our setting, to obtain the Lindbladian form of Eq.~\eqref{eq:Lindbladian}, let $G = \sum_{a\in\calA} G^a$ in Eq.~\eqref{eq:coherent_term_expression}, 
with 
\begin{equation}
    G^a = \intinf g(t) \: e^{iHt} \, ((L^a)^\dag L^a) \, e^{-iHt} \dd t.
\end{equation}
Letting
\begin{equation}
    \calL^a(X) = i[G^a, X] + (L^a)^\dag X L^a - \frac{1}{2} \left\{ (L^a)^\dag L^a, X \right\},
\end{equation} we obtain 
$\calL = \sum_{a\in\calA} \calL^a$. For a specific temperature $\beta$, to explicitly denote the temperature dependence, we write $\calL_\beta = \sum_{a\in\calA} \calL_\beta^a$.
This Lindbladian is $\sigma_\beta$-detailed balance, as shown in \cite{Ding2025EfficientQuantumGibbs}. We further assume that it is irreducible.

It follows from Eq.~\eqref{eq:superoperator_vectorization} that, for superoperator $K(\mathplaceholder) = A (\mathplaceholder) B^\dagger$, its vectorization $vKv^{-1} = A \otimes B^*$, where $B^*$ is the entrywise complex conjugate. 
Subsequently, our parent Hamiltonian at temperature $\beta$ can be written as 
\begin{equation}
    \sfH_\beta = \sum_{a\in\calA} v \calL_\beta^a v^{-1} = \sum_{a\in\calA} \sfH_{\beta}^a,
\end{equation}
where
\begin{equation}
\begin{split}
    \sfH_{\beta}^a = \Big (\sigma_\beta^{\frac{1}{4}} \otimes (\sigma_\beta^\top)^{\frac{1}{4}} \Big) \cdot 
    \Big ( 
    iG_{\beta}^a \otimes I - iI \otimes (G_{\beta}^a)^\top 
    \\ + (L_\beta^a)^\dag \otimes (L_{\beta}^a)^\top - \frac{1}{2} (L_{\beta}^a)^\dag L_{\beta}^a \otimes I -  \\ \frac{1}{2} I \otimes (L_{\beta}^a)^\top (L_{\beta}^a)^* \Big) 
    \cdot\Big( \sigma_\beta^{-\frac{1}{4}} \otimes (\sigma_\beta^\top)^{-\frac{1}{4}} \Big),
\end{split}
\end{equation}
where the first and last term is the vectorization of $\Gamma_\sigma$ defined in Eq.~\eqref{eq:gamma_sigma_def} and the middle term is the vectorization of $\calL_{\beta}^a$. The parent Hamiltonian is frustration-free, with $\ket{\sqrt{\sigma_\beta}}$ being the ground state for both $\sfH_\beta$ and $\sfH_\beta^a$.

Consider a commuting Hamiltonian $H = \sum_{m} H_m$ with locality $l$ and bounded degree $q$ and a set of coupling operators $\{A^a\}$ with locality $k_A$ and bounded degree $g_A$.
The support of $A^a(t)$ is the union of the support of $A^a$ and all Hamiltonian terms that intersect $A^a$. The number of such intersecting Hamiltonian terms is bounded. Then by locality, the support of $A^a(t)$ is bounded. By definition, $L^a_\beta$, $G^a_\beta$, both before and after applying by $\sigma_\beta^{\frac{1}{4}} \mathplaceholder \sigma_\beta^{-\frac{1}{4}} = e^{-\frac{1}{4}\beta H} \mathplaceholder e^{\frac{1}{4}\beta H}$,
have the same support as $A^a(t)$. Therefore, the parent Hamiltonian is local. Furthermore, since $A^a(t)$ has bounded support, the number of Hamiltonian terms that intersect $A^a(t)$ is bounded, therefore the number of $a'$ such that $\suppm(A^{a'}(t)) \cap \suppm( A^a(t) ) \neq \emptyset$ 
is bounded. Thus, the parent Hamiltonian is bounded-degree local.

Consider the annealing path with inverse temperature $\nu \in [0, \beta]$. Each $\calL_\nu$ and $\sfH_\nu$ are constructed as above, thus they meet the requirements for our method to apply. To summarize, we have the following corollary.

\begin{cor}
\label{cor:commuting}
    Consider commuting and bounded-degree local Hamiltonian $H$ (Definition~\ref{defn:bounded_degree_local}). For each $\nu\in [0, \beta]$, we construct $\sigma_\nu$-detailed balanced $\calL_\nu$ as above and assume its irreducibility. Then, the corresponding parent Hamiltonian $\sfH_\nu$ will be bounded-degree local and frustration-free, thus it satisfies the requirements of Theorem~\ref{thm:gibbsstateprepcomplex}. Therefore there exists a quantum algorithm that prepares the purified Gibbs state at inverse temperature $\beta$ with error $\delta$ and success probability at least $1-\delta$ with gate complexity
    \begin{equation}
        \calO \left( \frac{M\beta \norm{H}}{\sqrt{\gamma}} \log^2 \left(\frac{\beta\norm{H}}{\delta}\right) \log^c \left( \frac{M}{\sqrt{\gamma}} \frac{\beta\norm{H}}{\delta} 
        \right) \right)
    \end{equation} using $\calO (\log M) + 1$ resettable ancilla qubits, where $c$ is the exponent in the Solovay-Kitaev theorem, and $\gamma = \min_j \mathrm{gap}(\mathcal{L}_{\beta_j})$, $\beta_j = j\beta/K$, $K=\Theta(\beta\|H\|)$.
\end{cor}

\section{Discussion and outlook}

In this work, we use the detectability lemma to improve quantum Gibbs sampling. We propose a quantum algorithm for preparing the stationary state without simulating Lindbladian evolution, thereby avoiding the associated overhead. We also present a quantum algorithm for preparing the Gibbs state of commuting bounded-degree local Hamiltonians, achieving a quadratic improvement in its dependence on the Lindbladian spectral gap. As a by-product, we also obtain an efficient quantum algorithm for preparing the ground state of frustration-free bounded-degree local Hamiltonians.

For future work, an important direction would be to extend the quadratic improvement in spectral-gap dependence to a broader class of Gibbs-state preparation problems. A natural way to achieve this would be to develop a stronger version of the detectability lemma that applies to quasi-local interactions. Beyond its conceptual interest, such a result could also lead to an improved Gibbs-sampling algorithm along the lines of Corollary~\ref{cor:gate_complexity_approximate_proj}.

{One might also want to apply the Lieb-Robinson bound to truncate each quasi-local jump operator to a ball of radius $R$, apply the present algorithm to the resulting finite-range generator, and choose $R$ so that the truncation tail is below the target accuracy \cite{LiebRobinson1972finite,ChenKastoryanoBrandaoGilyen2023quantum,ChenKastoryanoGilyen2023efficient}. A direct application, however, goes beyond the scope of our current analysis. A generic truncation can destroy exact KMS detailed balance and the exact frustration freeness of the parent Hamiltonian, both of which enter our proof. Moreover, since local fixed-point projections are implemented via unitary synthesis in our current algorithm, such truncation leads to superpolynomial scaling with respect to the truncation error. How to handle the non-commuting case while preserving algorithmic efficiency in this detectability-lemma-based approach is therefore still an open problem.}

More broadly, it would be very interesting to understand whether the detectability lemma can be used as a tool for analyzing the mixing times of Lindbladians 
\cite{temme2014hypercontractivity,ding2024polynomial,bardet2024entropy,alicki2009thermalization,ramkumar2024mixing,KochanowskiAlhambraCapelRouze2024rapid,BardetCapelLiEtAl2023rapid,RouzeFrancaAlhambra2024optimal,capel2024quasi,kastoryano2013quantum,temme2015fast,ChenLiLuYing2024randomized,Rakovszky2024bottlenecks,LiLu2024quantum,Barthel2022solving,Fang2024mixing,DingLiLin2024efficient,DingChenLin2024single}.
Progress in this direction could be especially valuable if it were to yield a genuine quantum advantage. At present, however, many seemingly efficient quantum algorithms based on Lindbladian dynamics are challenged by strong classical algorithms \cite{BakshiLiuMoitraTang2024high,McDonoughYinLucasZhang2025lieb,RamkumarCaiTongJiang2025high,ChenRouzeEtAl2025convergence,ZlokapaKiani2026syk}. Establishing a framework in which detectability-lemma techniques lead to provable quantum advantages in this setting would therefore be particularly compelling.

Another important direction is improving Lindbladian simulation. The recently derived commutator bound for Lindbladians \cite{WangZhouEtAl2026lindbladian} yields a better dependence on the number of terms for local Lindbladians, but its precision dependence is worse than that of \cite{CleveWang2017}. Achieving the same spacetime-volume scaling for Lindbladian simulation as has been obtained for Hamiltonian simulation \cite{HaahHastingsKothariLow2021quantum,ChildsSuTranWiebeZhu2021theory} remains a major open problem.

\section*{Acknowledgments}
Y.T. thanks Andr{\'a}s Gily{\'e}n for helpful discussion.
C.Z. acknowledges support from NSF QLCI grant OMA-2120757.
J.F. acknowledges support by the National Science Foundation under awards DMS-2309378
D.F. acknowledges the support from the U.S. Department of Energy, Office of Science, Accelerated Research in Quantum Computing Centers, Quantum Utility through Advanced Computational Quantum Algorithms, grant no. DE-SC0025572, and National Science Foundation via the NSF CAREER award DMS-2438074.

\bibliography{ref,ref2}

\begin{thebibliography}{71}%
\makeatletter
\providecommand \@ifxundefined [1]{%
 \@ifx{#1\undefined}
}%
\providecommand \@ifnum [1]{%
 \ifnum #1\expandafter \@firstoftwo
 \else \expandafter \@secondoftwo
 \fi
}%
\providecommand \@ifx [1]{%
 \ifx #1\expandafter \@firstoftwo
 \else \expandafter \@secondoftwo
 \fi
}%
\providecommand \natexlab [1]{#1}%
\providecommand \enquote  [1]{``#1''}%
\providecommand \bibnamefont  [1]{#1}%
\providecommand \bibfnamefont [1]{#1}%
\providecommand \citenamefont [1]{#1}%
\providecommand \href@noop [0]{\@secondoftwo}%
\providecommand \href [0]{\begingroup \@sanitize@url \@href}%
\providecommand \@href[1]{\@@startlink{#1}\@@href}%
\providecommand \@@href[1]{\endgroup#1\@@endlink}%
\providecommand \@sanitize@url [0]{\catcode `\\12\catcode `\$12\catcode `\&12\catcode `\#12\catcode `\^12\catcode `\_12\catcode `\%12\relax}%
\providecommand \@@startlink[1]{}%
\providecommand \@@endlink[0]{}%
\providecommand \url  [0]{\begingroup\@sanitize@url \@url }%
\providecommand \@url [1]{\endgroup\@href {#1}{\urlprefix }}%
\providecommand \urlprefix  [0]{URL }%
\providecommand \Eprint [0]{\href }%
\providecommand \doibase [0]{https://doi.org/}%
\providecommand \selectlanguage [0]{\@gobble}%
\providecommand \bibinfo  [0]{\@secondoftwo}%
\providecommand \bibfield  [0]{\@secondoftwo}%
\providecommand \translation [1]{[#1]}%
\providecommand \BibitemOpen [0]{}%
\providecommand \bibitemStop [0]{}%
\providecommand \bibitemNoStop [0]{.\EOS\space}%
\providecommand \EOS [0]{\spacefactor3000\relax}%
\providecommand \BibitemShut  [1]{\csname bibitem#1\endcsname}%
\let\auto@bib@innerbib\@empty
\bibitem [{\citenamefont {Brand{\~a}o}\ \emph {et~al.}(2017)\citenamefont {Brand{\~a}o}, \citenamefont {Kalev}, \citenamefont {Li}, \citenamefont {Lin}, \citenamefont {Svore},\ and\ \citenamefont {Wu}}]{BrandaoKalev2017quantum}%
  \BibitemOpen
  \bibfield  {author} {\bibinfo {author} {\bibfnamefont {F.~G.}\ \bibnamefont {Brand{\~a}o}}, \bibinfo {author} {\bibfnamefont {A.}~\bibnamefont {Kalev}}, \bibinfo {author} {\bibfnamefont {T.}~\bibnamefont {Li}}, \bibinfo {author} {\bibfnamefont {C.~Y.-Y.}\ \bibnamefont {Lin}}, \bibinfo {author} {\bibfnamefont {K.~M.}\ \bibnamefont {Svore}},\ and\ \bibinfo {author} {\bibfnamefont {X.}~\bibnamefont {Wu}},\ }\bibfield  {title} {\bibinfo {title} {Quantum sdp solvers: Large speed-ups, optimality, and applications to quantum learning},\ }\href@noop {} {\bibfield  {journal} {\bibinfo  {journal} {arXiv preprint arXiv:1710.02581}\ } (\bibinfo {year} {2017})}\BibitemShut {NoStop}%
\bibitem [{\citenamefont {Van~Apeldoorn}\ \emph {et~al.}(2017)\citenamefont {Van~Apeldoorn}, \citenamefont {Gily{\'e}n}, \citenamefont {Gribling},\ and\ \citenamefont {de~Wolf}}]{vanApeldoorn2017quantum}%
  \BibitemOpen
  \bibfield  {author} {\bibinfo {author} {\bibfnamefont {J.}~\bibnamefont {Van~Apeldoorn}}, \bibinfo {author} {\bibfnamefont {A.}~\bibnamefont {Gily{\'e}n}}, \bibinfo {author} {\bibfnamefont {S.}~\bibnamefont {Gribling}},\ and\ \bibinfo {author} {\bibfnamefont {R.}~\bibnamefont {de~Wolf}},\ }\bibfield  {title} {\bibinfo {title} {Quantum sdp-solvers: Better upper and lower bounds},\ }in\ \href@noop {} {\emph {\bibinfo {booktitle} {2017 IEEE 58th Annual Symposium on Foundations of Computer Science (FOCS)}}}\ (\bibinfo {organization} {IEEE},\ \bibinfo {year} {2017})\ pp.\ \bibinfo {pages} {403--414}\BibitemShut {NoStop}%
\bibitem [{\citenamefont {Motta}\ \emph {et~al.}(2020)\citenamefont {Motta}, \citenamefont {Sun}, \citenamefont {Tan}, \citenamefont {O’Rourke}, \citenamefont {Ye}, \citenamefont {Minnich}, \citenamefont {Brandao},\ and\ \citenamefont {Chan}}]{Motta2020determining}%
  \BibitemOpen
  \bibfield  {author} {\bibinfo {author} {\bibfnamefont {M.}~\bibnamefont {Motta}}, \bibinfo {author} {\bibfnamefont {C.}~\bibnamefont {Sun}}, \bibinfo {author} {\bibfnamefont {A.~T.}\ \bibnamefont {Tan}}, \bibinfo {author} {\bibfnamefont {M.~J.}\ \bibnamefont {O’Rourke}}, \bibinfo {author} {\bibfnamefont {E.}~\bibnamefont {Ye}}, \bibinfo {author} {\bibfnamefont {A.~J.}\ \bibnamefont {Minnich}}, \bibinfo {author} {\bibfnamefont {F.~G.}\ \bibnamefont {Brandao}},\ and\ \bibinfo {author} {\bibfnamefont {G.~K.-L.}\ \bibnamefont {Chan}},\ }\bibfield  {title} {\bibinfo {title} {Determining eigenstates and thermal states on a quantum computer using quantum imaginary time evolution},\ }\href@noop {} {\bibfield  {journal} {\bibinfo  {journal} {Nature Physics}\ }\textbf {\bibinfo {volume} {16}},\ \bibinfo {pages} {205} (\bibinfo {year} {2020})}\BibitemShut {NoStop}%
\bibitem [{\citenamefont {Chen}\ \emph {et~al.}(2023{\natexlab{a}})\citenamefont {Chen}, \citenamefont {Kastoryano}, \citenamefont {Brand{\~a}o},\ and\ \citenamefont {Gily{\'e}n}}]{chen2023quantumthermalstatepreparation}%
  \BibitemOpen
  \bibfield  {author} {\bibinfo {author} {\bibfnamefont {C.-F.}\ \bibnamefont {Chen}}, \bibinfo {author} {\bibfnamefont {M.~J.}\ \bibnamefont {Kastoryano}}, \bibinfo {author} {\bibfnamefont {F.~G.}\ \bibnamefont {Brand{\~a}o}},\ and\ \bibinfo {author} {\bibfnamefont {A.}~\bibnamefont {Gily{\'e}n}},\ }\bibfield  {title} {\bibinfo {title} {Quantum thermal state preparation},\ }\href@noop {} {\bibfield  {journal} {\bibinfo  {journal} {arXiv preprint arXiv:2303.18224}\ } (\bibinfo {year} {2023}{\natexlab{a}})}\BibitemShut {NoStop}%
\bibitem [{\citenamefont {Lloyd}\ and\ \citenamefont {Abanin}(2025)}]{LloydAbanin2025quantum}%
  \BibitemOpen
  \bibfield  {author} {\bibinfo {author} {\bibfnamefont {J.}~\bibnamefont {Lloyd}}\ and\ \bibinfo {author} {\bibfnamefont {D.~A.}\ \bibnamefont {Abanin}},\ }\bibfield  {title} {\bibinfo {title} {Quantum thermal state preparation for near-term quantum processors},\ }\href@noop {} {\bibfield  {journal} {\bibinfo  {journal} {arXiv preprint arXiv:2506.21318}\ } (\bibinfo {year} {2025})}\BibitemShut {NoStop}%
\bibitem [{\citenamefont {Chen}\ \emph {et~al.}(2023{\natexlab{b}})\citenamefont {Chen}, \citenamefont {Kastoryano}, \citenamefont {Brandao},\ and\ \citenamefont {Gilyen}}]{ChenKastoryanoBrandaoGilyen2023}%
  \BibitemOpen
  \bibfield  {author} {\bibinfo {author} {\bibfnamefont {C.-F.}\ \bibnamefont {Chen}}, \bibinfo {author} {\bibfnamefont {M.~J.}\ \bibnamefont {Kastoryano}}, \bibinfo {author} {\bibfnamefont {F.~G. S.~L.}\ \bibnamefont {Brandao}},\ and\ \bibinfo {author} {\bibfnamefont {A.}~\bibnamefont {Gilyen}},\ }\href@noop {} {\bibinfo {title} {Quantum thermal state preparation}} (\bibinfo {year} {2023}{\natexlab{b}}),\ \Eprint {https://arxiv.org/abs/2303.18224} {arXiv:2303.18224 [quant-ph]} \BibitemShut {NoStop}%
\bibitem [{\citenamefont {Rall}\ \emph {et~al.}(2023)\citenamefont {Rall}, \citenamefont {Wang},\ and\ \citenamefont {Wocjan}}]{RallWangWocjan2023}%
  \BibitemOpen
  \bibfield  {author} {\bibinfo {author} {\bibfnamefont {P.}~\bibnamefont {Rall}}, \bibinfo {author} {\bibfnamefont {C.}~\bibnamefont {Wang}},\ and\ \bibinfo {author} {\bibfnamefont {P.}~\bibnamefont {Wocjan}},\ }\bibfield  {title} {\bibinfo {title} {Thermal {S}tate {P}reparation via {R}ounding {P}romises},\ }\href {https://doi.org/10.22331/q-2023-10-10-1132} {\bibfield  {journal} {\bibinfo  {journal} {{Quantum}}\ }\textbf {\bibinfo {volume} {7}},\ \bibinfo {pages} {1132} (\bibinfo {year} {2023})}\BibitemShut {NoStop}%
\bibitem [{\citenamefont {Ding}\ \emph {et~al.}(2025{\natexlab{a}})\citenamefont {Ding}, \citenamefont {Li},\ and\ \citenamefont {Lin}}]{DingLiLin2024efficient}%
  \BibitemOpen
  \bibfield  {author} {\bibinfo {author} {\bibfnamefont {Z.}~\bibnamefont {Ding}}, \bibinfo {author} {\bibfnamefont {B.}~\bibnamefont {Li}},\ and\ \bibinfo {author} {\bibfnamefont {L.}~\bibnamefont {Lin}},\ }\bibfield  {title} {\bibinfo {title} {Efficient quantum {Gibbs} samplers with {Kubo--Martin--Schwinger} detailed balance condition},\ }\href@noop {} {\bibfield  {journal} {\bibinfo  {journal} {Communications in Mathematical Physics}\ }\textbf {\bibinfo {volume} {406}},\ \bibinfo {pages} {67} (\bibinfo {year} {2025}{\natexlab{a}})}\BibitemShut {NoStop}%
\bibitem [{\citenamefont {Cleve}\ and\ \citenamefont {Wang}(2017)}]{CleveWang2017}%
  \BibitemOpen
  \bibfield  {author} {\bibinfo {author} {\bibfnamefont {R.}~\bibnamefont {Cleve}}\ and\ \bibinfo {author} {\bibfnamefont {C.}~\bibnamefont {Wang}},\ }\bibfield  {title} {\bibinfo {title} {Efficient quantum algorithms for simulating lindblad evolution},\ }in\ \href@noop {} {\emph {\bibinfo {booktitle} {44th International Colloquium on Automata, Languages, and Programming (ICALP 2017)}}}\ (\bibinfo {organization} {Schloss-Dagstuhl-Leibniz Zentrum f{\"u}r Informatik},\ \bibinfo {year} {2017})\BibitemShut {NoStop}%
\bibitem [{\citenamefont {Wocjan}\ and\ \citenamefont {Temme}(2023)}]{WocjanTemme2023}%
  \BibitemOpen
  \bibfield  {author} {\bibinfo {author} {\bibfnamefont {P.}~\bibnamefont {Wocjan}}\ and\ \bibinfo {author} {\bibfnamefont {K.}~\bibnamefont {Temme}},\ }\bibfield  {title} {\bibinfo {title} {Szegedy walk unitaries for quantum maps},\ }\href@noop {} {\bibfield  {journal} {\bibinfo  {journal} {Communications in Mathematical Physics}\ }\textbf {\bibinfo {volume} {402}},\ \bibinfo {pages} {3201} (\bibinfo {year} {2023})}\BibitemShut {NoStop}%
\bibitem [{\citenamefont {Chen}\ \emph {et~al.}(2023{\natexlab{c}})\citenamefont {Chen}, \citenamefont {Kastoryano},\ and\ \citenamefont {Gilyen}}]{ChenKastoryanoGilyen2023}%
  \BibitemOpen
  \bibfield  {author} {\bibinfo {author} {\bibfnamefont {C.-F.}\ \bibnamefont {Chen}}, \bibinfo {author} {\bibfnamefont {M.~J.}\ \bibnamefont {Kastoryano}},\ and\ \bibinfo {author} {\bibfnamefont {A.}~\bibnamefont {Gilyen}},\ }\href@noop {} {\bibinfo {title} {An efficient and exact noncommutative quantum gibbs sampler}} (\bibinfo {year} {2023}{\natexlab{c}}),\ \Eprint {https://arxiv.org/abs/2311.09207} {arXiv:2311.09207 [quant-ph]} \BibitemShut {NoStop}%
\bibitem [{\citenamefont {Li}\ and\ \citenamefont {Wang}(2023)}]{LiWang2023}%
  \BibitemOpen
  \bibfield  {author} {\bibinfo {author} {\bibfnamefont {X.}~\bibnamefont {Li}}\ and\ \bibinfo {author} {\bibfnamefont {C.}~\bibnamefont {Wang}},\ }\bibfield  {title} {\bibinfo {title} {Simulating markovian open quantum systems using higher-order series expansion},\ }in\ \href@noop {} {\emph {\bibinfo {booktitle} {50th International Colloquium on Automata, Languages, and Programming (ICALP 2023)}}}\ (\bibinfo {organization} {Schloss-Dagstuhl-Leibniz Zentrum f{\"u}r Informatik},\ \bibinfo {year} {2023})\BibitemShut {NoStop}%
\bibitem [{\citenamefont {Ding}\ \emph {et~al.}(2024{\natexlab{a}})\citenamefont {Ding}, \citenamefont {Li},\ and\ \citenamefont {Lin}}]{DingLiLin2024simulating}%
  \BibitemOpen
  \bibfield  {author} {\bibinfo {author} {\bibfnamefont {Z.}~\bibnamefont {Ding}}, \bibinfo {author} {\bibfnamefont {X.}~\bibnamefont {Li}},\ and\ \bibinfo {author} {\bibfnamefont {L.}~\bibnamefont {Lin}},\ }\bibfield  {title} {\bibinfo {title} {Simulating open quantum systems using hamiltonian simulations},\ }\href@noop {} {\bibfield  {journal} {\bibinfo  {journal} {PRX quantum}\ }\textbf {\bibinfo {volume} {5}},\ \bibinfo {pages} {020332} (\bibinfo {year} {2024}{\natexlab{a}})}\BibitemShut {NoStop}%
\bibitem [{\citenamefont {Pocrnic}\ \emph {et~al.}(2024)\citenamefont {Pocrnic}, \citenamefont {Segal},\ and\ \citenamefont {Wiebe}}]{PocrnicSegalWiebe2024}%
  \BibitemOpen
  \bibfield  {author} {\bibinfo {author} {\bibfnamefont {M.}~\bibnamefont {Pocrnic}}, \bibinfo {author} {\bibfnamefont {D.}~\bibnamefont {Segal}},\ and\ \bibinfo {author} {\bibfnamefont {N.}~\bibnamefont {Wiebe}},\ }\href@noop {} {\bibinfo {title} {Quantum simulation of lindbladian dynamics via repeated interactions}} (\bibinfo {year} {2024}),\ \Eprint {https://arxiv.org/abs/2312.05371} {arXiv:2312.05371 [quant-ph]} \BibitemShut {NoStop}%
\bibitem [{\citenamefont {Chen}\ \emph {et~al.}(2025{\natexlab{a}})\citenamefont {Chen}, \citenamefont {Kastoryano}, \citenamefont {Brand{\~a}o},\ and\ \citenamefont {Gily{\'e}n}}]{chen2025efficient}%
  \BibitemOpen
  \bibfield  {author} {\bibinfo {author} {\bibfnamefont {C.-F.}\ \bibnamefont {Chen}}, \bibinfo {author} {\bibfnamefont {M.}~\bibnamefont {Kastoryano}}, \bibinfo {author} {\bibfnamefont {F.~G.}\ \bibnamefont {Brand{\~a}o}},\ and\ \bibinfo {author} {\bibfnamefont {A.}~\bibnamefont {Gily{\'e}n}},\ }\bibfield  {title} {\bibinfo {title} {Efficient quantum thermal simulation},\ }\href@noop {} {\bibfield  {journal} {\bibinfo  {journal} {Nature}\ }\textbf {\bibinfo {volume} {646}},\ \bibinfo {pages} {561} (\bibinfo {year} {2025}{\natexlab{a}})}\BibitemShut {NoStop}%
\bibitem [{\citenamefont {Temme}\ \emph {et~al.}(2014)\citenamefont {Temme}, \citenamefont {Kastoryano}, \citenamefont {Ruskai}, \citenamefont {Wolf},\ and\ \citenamefont {Verstraete}}]{temme2014hypercontractivity}%
  \BibitemOpen
  \bibfield  {author} {\bibinfo {author} {\bibfnamefont {K.}~\bibnamefont {Temme}}, \bibinfo {author} {\bibfnamefont {M.~J.}\ \bibnamefont {Kastoryano}}, \bibinfo {author} {\bibfnamefont {M.~B.}\ \bibnamefont {Ruskai}}, \bibinfo {author} {\bibfnamefont {M.~M.}\ \bibnamefont {Wolf}},\ and\ \bibinfo {author} {\bibfnamefont {F.}~\bibnamefont {Verstraete}},\ }\bibfield  {title} {\bibinfo {title} {Hypercontractivity of quasi-free quantum semigroups},\ }\href@noop {} {\bibfield  {journal} {\bibinfo  {journal} {Journal of Mathematical Physics}\ }\textbf {\bibinfo {volume} {55}} (\bibinfo {year} {2014})}\BibitemShut {NoStop}%
\bibitem [{\citenamefont {Ding}\ \emph {et~al.}(2024{\natexlab{b}})\citenamefont {Ding}, \citenamefont {Li}, \citenamefont {Lin},\ and\ \citenamefont {Zhang}}]{ding2024polynomial}%
  \BibitemOpen
  \bibfield  {author} {\bibinfo {author} {\bibfnamefont {Z.}~\bibnamefont {Ding}}, \bibinfo {author} {\bibfnamefont {B.}~\bibnamefont {Li}}, \bibinfo {author} {\bibfnamefont {L.}~\bibnamefont {Lin}},\ and\ \bibinfo {author} {\bibfnamefont {R.}~\bibnamefont {Zhang}},\ }\bibfield  {title} {\bibinfo {title} {Polynomial-time preparation of low-temperature {G}ibbs states for 2d toric code},\ }\href@noop {} {\bibfield  {journal} {\bibinfo  {journal} {arXiv preprint arXiv:2410.01206}\ } (\bibinfo {year} {2024}{\natexlab{b}})}\BibitemShut {NoStop}%
\bibitem [{\citenamefont {Bardet}\ \emph {et~al.}(2024)\citenamefont {Bardet}, \citenamefont {Capel}, \citenamefont {Gao}, \citenamefont {Lucia}, \citenamefont {P{\'e}rez-Garc{\'\i}a},\ and\ \citenamefont {Rouz{\'e}}}]{bardet2024entropy}%
  \BibitemOpen
  \bibfield  {author} {\bibinfo {author} {\bibfnamefont {I.}~\bibnamefont {Bardet}}, \bibinfo {author} {\bibfnamefont {{\'A}.}~\bibnamefont {Capel}}, \bibinfo {author} {\bibfnamefont {L.}~\bibnamefont {Gao}}, \bibinfo {author} {\bibfnamefont {A.}~\bibnamefont {Lucia}}, \bibinfo {author} {\bibfnamefont {D.}~\bibnamefont {P{\'e}rez-Garc{\'\i}a}},\ and\ \bibinfo {author} {\bibfnamefont {C.}~\bibnamefont {Rouz{\'e}}},\ }\bibfield  {title} {\bibinfo {title} {Entropy decay for davies semigroups of a one dimensional quantum lattice},\ }\href@noop {} {\bibfield  {journal} {\bibinfo  {journal} {Communications in Mathematical Physics}\ }\textbf {\bibinfo {volume} {405}},\ \bibinfo {pages} {42} (\bibinfo {year} {2024})}\BibitemShut {NoStop}%
\bibitem [{\citenamefont {Alicki}\ \emph {et~al.}(2009)\citenamefont {Alicki}, \citenamefont {Fannes},\ and\ \citenamefont {Horodecki}}]{alicki2009thermalization}%
  \BibitemOpen
  \bibfield  {author} {\bibinfo {author} {\bibfnamefont {R.}~\bibnamefont {Alicki}}, \bibinfo {author} {\bibfnamefont {M.}~\bibnamefont {Fannes}},\ and\ \bibinfo {author} {\bibfnamefont {M.}~\bibnamefont {Horodecki}},\ }\bibfield  {title} {\bibinfo {title} {On thermalization in {K}itaev's 2d model},\ }\href@noop {} {\bibfield  {journal} {\bibinfo  {journal} {Journal of Physics A: Mathematical and Theoretical}\ }\textbf {\bibinfo {volume} {42}},\ \bibinfo {pages} {065303} (\bibinfo {year} {2009})}\BibitemShut {NoStop}%
\bibitem [{\citenamefont {Ramkumar}\ and\ \citenamefont {Soleimanifar}(2024)}]{ramkumar2024mixing}%
  \BibitemOpen
  \bibfield  {author} {\bibinfo {author} {\bibfnamefont {A.}~\bibnamefont {Ramkumar}}\ and\ \bibinfo {author} {\bibfnamefont {M.}~\bibnamefont {Soleimanifar}},\ }\bibfield  {title} {\bibinfo {title} {Mixing time of quantum {G}ibbs sampling for random sparse {H}amiltonians},\ }\href@noop {} {\bibfield  {journal} {\bibinfo  {journal} {arXiv preprint arXiv:2411.04454}\ } (\bibinfo {year} {2024})}\BibitemShut {NoStop}%
\bibitem [{\citenamefont {Kochanowski}\ \emph {et~al.}(2024)\citenamefont {Kochanowski}, \citenamefont {Alhambra}, \citenamefont {Capel},\ and\ \citenamefont {Rouz{\'e}}}]{KochanowskiAlhambraCapelRouze2024rapid}%
  \BibitemOpen
  \bibfield  {author} {\bibinfo {author} {\bibfnamefont {J.}~\bibnamefont {Kochanowski}}, \bibinfo {author} {\bibfnamefont {A.~M.}\ \bibnamefont {Alhambra}}, \bibinfo {author} {\bibfnamefont {A.}~\bibnamefont {Capel}},\ and\ \bibinfo {author} {\bibfnamefont {C.}~\bibnamefont {Rouz{\'e}}},\ }\bibfield  {title} {\bibinfo {title} {Rapid thermalization of dissipative many-body dynamics of commuting {H}amiltonians},\ }\href@noop {} {\bibfield  {journal} {\bibinfo  {journal} {arXiv preprint arXiv:2404.16780}\ } (\bibinfo {year} {2024})}\BibitemShut {NoStop}%
\bibitem [{\citenamefont {Bardet}\ \emph {et~al.}(2023)\citenamefont {Bardet}, \citenamefont {Capel}, \citenamefont {Gao}, \citenamefont {Lucia}, \citenamefont {P{\'e}rez-Garc{\'\i}a},\ and\ \citenamefont {Rouz{\'e}}}]{BardetCapelLiEtAl2023rapid}%
  \BibitemOpen
  \bibfield  {author} {\bibinfo {author} {\bibfnamefont {I.}~\bibnamefont {Bardet}}, \bibinfo {author} {\bibfnamefont {{\'A}.}~\bibnamefont {Capel}}, \bibinfo {author} {\bibfnamefont {L.}~\bibnamefont {Gao}}, \bibinfo {author} {\bibfnamefont {A.}~\bibnamefont {Lucia}}, \bibinfo {author} {\bibfnamefont {D.}~\bibnamefont {P{\'e}rez-Garc{\'\i}a}},\ and\ \bibinfo {author} {\bibfnamefont {C.}~\bibnamefont {Rouz{\'e}}},\ }\bibfield  {title} {\bibinfo {title} {Rapid thermalization of spin chain commuting {H}amiltonians},\ }\href@noop {} {\bibfield  {journal} {\bibinfo  {journal} {Physical Review Letters}\ }\textbf {\bibinfo {volume} {130}},\ \bibinfo {pages} {060401} (\bibinfo {year} {2023})}\BibitemShut {NoStop}%
\bibitem [{\citenamefont {Rouz{\'e}}\ \emph {et~al.}(2024)\citenamefont {Rouz{\'e}}, \citenamefont {Fran{\c{c}}a},\ and\ \citenamefont {Alhambra}}]{RouzeFrancaAlhambra2024optimal}%
  \BibitemOpen
  \bibfield  {author} {\bibinfo {author} {\bibfnamefont {C.}~\bibnamefont {Rouz{\'e}}}, \bibinfo {author} {\bibfnamefont {D.~S.}\ \bibnamefont {Fran{\c{c}}a}},\ and\ \bibinfo {author} {\bibfnamefont {{\'A}.~M.}\ \bibnamefont {Alhambra}},\ }\bibfield  {title} {\bibinfo {title} {Optimal quantum algorithm for {G}ibbs state preparation},\ }\href@noop {} {\bibfield  {journal} {\bibinfo  {journal} {arXiv preprint arXiv:2411.04885}\ } (\bibinfo {year} {2024})}\BibitemShut {NoStop}%
\bibitem [{\citenamefont {Capel}\ \emph {et~al.}(2024)\citenamefont {Capel}, \citenamefont {Gondolf}, \citenamefont {Kochanowski},\ and\ \citenamefont {Rouz{\'e}}}]{capel2024quasi}%
  \BibitemOpen
  \bibfield  {author} {\bibinfo {author} {\bibfnamefont {{\'A}.}~\bibnamefont {Capel}}, \bibinfo {author} {\bibfnamefont {P.}~\bibnamefont {Gondolf}}, \bibinfo {author} {\bibfnamefont {J.}~\bibnamefont {Kochanowski}},\ and\ \bibinfo {author} {\bibfnamefont {C.}~\bibnamefont {Rouz{\'e}}},\ }\bibfield  {title} {\bibinfo {title} {Quasi-optimal sampling from {G}ibbs states via non-commutative optimal transport metrics},\ }\href@noop {} {\bibfield  {journal} {\bibinfo  {journal} {arXiv preprint arXiv:2412.01732}\ } (\bibinfo {year} {2024})}\BibitemShut {NoStop}%
\bibitem [{\citenamefont {Kastoryano}\ and\ \citenamefont {Temme}(2013)}]{kastoryano2013quantum}%
  \BibitemOpen
  \bibfield  {author} {\bibinfo {author} {\bibfnamefont {M.~J.}\ \bibnamefont {Kastoryano}}\ and\ \bibinfo {author} {\bibfnamefont {K.}~\bibnamefont {Temme}},\ }\bibfield  {title} {\bibinfo {title} {Quantum logarithmic sobolev inequalities and rapid mixing},\ }\href@noop {} {\bibfield  {journal} {\bibinfo  {journal} {Journal of Mathematical Physics}\ }\textbf {\bibinfo {volume} {54}} (\bibinfo {year} {2013})}\BibitemShut {NoStop}%
\bibitem [{\citenamefont {Temme}\ and\ \citenamefont {Kastoryano}(2015)}]{temme2015fast}%
  \BibitemOpen
  \bibfield  {author} {\bibinfo {author} {\bibfnamefont {K.}~\bibnamefont {Temme}}\ and\ \bibinfo {author} {\bibfnamefont {M.~J.}\ \bibnamefont {Kastoryano}},\ }\bibfield  {title} {\bibinfo {title} {How fast do stabilizer {H}amiltonians thermalize?},\ }\href@noop {} {\bibfield  {journal} {\bibinfo  {journal} {arXiv preprint arXiv:1505.07811}\ } (\bibinfo {year} {2015})}\BibitemShut {NoStop}%
\bibitem [{\citenamefont {Chen}\ \emph {et~al.}(2024)\citenamefont {Chen}, \citenamefont {Li}, \citenamefont {Lu},\ and\ \citenamefont {Ying}}]{ChenLiLuYing2024randomized}%
  \BibitemOpen
  \bibfield  {author} {\bibinfo {author} {\bibfnamefont {H.}~\bibnamefont {Chen}}, \bibinfo {author} {\bibfnamefont {B.}~\bibnamefont {Li}}, \bibinfo {author} {\bibfnamefont {J.}~\bibnamefont {Lu}},\ and\ \bibinfo {author} {\bibfnamefont {L.}~\bibnamefont {Ying}},\ }\bibfield  {title} {\bibinfo {title} {A randomized method for simulating {L}indblad equations and thermal state preparation},\ }\href@noop {} {\bibfield  {journal} {\bibinfo  {journal} {arXiv preprint arXiv:2407.06594}\ } (\bibinfo {year} {2024})}\BibitemShut {NoStop}%
\bibitem [{\citenamefont {Rakovszky}\ \emph {et~al.}(2024)\citenamefont {Rakovszky}, \citenamefont {Placke}, \citenamefont {Breuckmann},\ and\ \citenamefont {Khemani}}]{Rakovszky2024bottlenecks}%
  \BibitemOpen
  \bibfield  {author} {\bibinfo {author} {\bibfnamefont {T.}~\bibnamefont {Rakovszky}}, \bibinfo {author} {\bibfnamefont {B.}~\bibnamefont {Placke}}, \bibinfo {author} {\bibfnamefont {N.~P.}\ \bibnamefont {Breuckmann}},\ and\ \bibinfo {author} {\bibfnamefont {V.}~\bibnamefont {Khemani}},\ }\bibfield  {title} {\bibinfo {title} {Bottlenecks in quantum channels and finite temperature phases of matter},\ }\href@noop {} {\bibfield  {journal} {\bibinfo  {journal} {arXiv preprint arXiv:2412.09598}\ } (\bibinfo {year} {2024})}\BibitemShut {NoStop}%
\bibitem [{\citenamefont {Li}\ and\ \citenamefont {Lu}(2024)}]{LiLu2024quantum}%
  \BibitemOpen
  \bibfield  {author} {\bibinfo {author} {\bibfnamefont {B.}~\bibnamefont {Li}}\ and\ \bibinfo {author} {\bibfnamefont {J.}~\bibnamefont {Lu}},\ }\bibfield  {title} {\bibinfo {title} {Quantum space-time {P}oincare inequality for {L}indblad dynamics},\ }\href@noop {} {\bibfield  {journal} {\bibinfo  {journal} {arXiv preprint arXiv:2406.09115}\ } (\bibinfo {year} {2024})}\BibitemShut {NoStop}%
\bibitem [{\citenamefont {Barthel}\ and\ \citenamefont {Zhang}(2022)}]{Barthel2022solving}%
  \BibitemOpen
  \bibfield  {author} {\bibinfo {author} {\bibfnamefont {T.}~\bibnamefont {Barthel}}\ and\ \bibinfo {author} {\bibfnamefont {Y.}~\bibnamefont {Zhang}},\ }\bibfield  {title} {\bibinfo {title} {Solving quasi-free and quadratic {L}indblad master equations for open fermionic and bosonic systems},\ }\href@noop {} {\bibfield  {journal} {\bibinfo  {journal} {Journal of Statistical Mechanics: Theory and Experiment}\ }\textbf {\bibinfo {volume} {2022}},\ \bibinfo {pages} {113101} (\bibinfo {year} {2022})}\BibitemShut {NoStop}%
\bibitem [{\citenamefont {Fang}\ \emph {et~al.}(2024)\citenamefont {Fang}, \citenamefont {Lu},\ and\ \citenamefont {Tong}}]{Fang2024mixing}%
  \BibitemOpen
  \bibfield  {author} {\bibinfo {author} {\bibfnamefont {D.}~\bibnamefont {Fang}}, \bibinfo {author} {\bibfnamefont {J.}~\bibnamefont {Lu}},\ and\ \bibinfo {author} {\bibfnamefont {Y.}~\bibnamefont {Tong}},\ }\bibfield  {title} {\bibinfo {title} {Mixing time of open quantum systems via hypocoercivity},\ }\href@noop {} {\bibfield  {journal} {\bibinfo  {journal} {arXiv preprint arXiv:2404.11503}\ } (\bibinfo {year} {2024})}\BibitemShut {NoStop}%
\bibitem [{\citenamefont {Ding}\ \emph {et~al.}(2024{\natexlab{c}})\citenamefont {Ding}, \citenamefont {Chen},\ and\ \citenamefont {Lin}}]{DingChenLin2024single}%
  \BibitemOpen
  \bibfield  {author} {\bibinfo {author} {\bibfnamefont {Z.}~\bibnamefont {Ding}}, \bibinfo {author} {\bibfnamefont {C.-F.}\ \bibnamefont {Chen}},\ and\ \bibinfo {author} {\bibfnamefont {L.}~\bibnamefont {Lin}},\ }\bibfield  {title} {\bibinfo {title} {Single-ancilla ground state preparation via lindbladians},\ }\href@noop {} {\bibfield  {journal} {\bibinfo  {journal} {Physical Review Research}\ }\textbf {\bibinfo {volume} {6}},\ \bibinfo {pages} {033147} (\bibinfo {year} {2024}{\natexlab{c}})}\BibitemShut {NoStop}%
\bibitem [{\citenamefont {Tong}\ and\ \citenamefont {Zhan}(2025)}]{TongZhan2025fast}%
  \BibitemOpen
  \bibfield  {author} {\bibinfo {author} {\bibfnamefont {Y.}~\bibnamefont {Tong}}\ and\ \bibinfo {author} {\bibfnamefont {Y.}~\bibnamefont {Zhan}},\ }\bibfield  {title} {\bibinfo {title} {Fast mixing of weakly interacting fermionic systems at any temperature},\ }\href@noop {} {\bibfield  {journal} {\bibinfo  {journal} {PRX Quantum}\ }\textbf {\bibinfo {volume} {6}},\ \bibinfo {pages} {030301} (\bibinfo {year} {2025})}\BibitemShut {NoStop}%
\bibitem [{\citenamefont {Bakshi}\ \emph {et~al.}(2025)\citenamefont {Bakshi}, \citenamefont {Liu}, \citenamefont {Moitra},\ and\ \citenamefont {Tang}}]{BakshiLiuMoitraTang2025dobrushin}%
  \BibitemOpen
  \bibfield  {author} {\bibinfo {author} {\bibfnamefont {A.}~\bibnamefont {Bakshi}}, \bibinfo {author} {\bibfnamefont {A.}~\bibnamefont {Liu}}, \bibinfo {author} {\bibfnamefont {A.}~\bibnamefont {Moitra}},\ and\ \bibinfo {author} {\bibfnamefont {E.}~\bibnamefont {Tang}},\ }\bibfield  {title} {\bibinfo {title} {A dobrushin condition for quantum markov chains: Rapid mixing and conditional mutual information at high temperature},\ }\href@noop {} {\bibfield  {journal} {\bibinfo  {journal} {arXiv preprint arXiv:2510.08542}\ } (\bibinfo {year} {2025})}\BibitemShut {NoStop}%
\bibitem [{\citenamefont {Gily{\'e}n}\ \emph {et~al.}(2024)\citenamefont {Gily{\'e}n}, \citenamefont {Chen}, \citenamefont {Doriguello},\ and\ \citenamefont {Kastoryano}}]{GilyenChenDoriguelloKastoryano2024quantum}%
  \BibitemOpen
  \bibfield  {author} {\bibinfo {author} {\bibfnamefont {A.}~\bibnamefont {Gily{\'e}n}}, \bibinfo {author} {\bibfnamefont {C.-F.}\ \bibnamefont {Chen}}, \bibinfo {author} {\bibfnamefont {J.~F.}\ \bibnamefont {Doriguello}},\ and\ \bibinfo {author} {\bibfnamefont {M.~J.}\ \bibnamefont {Kastoryano}},\ }\bibfield  {title} {\bibinfo {title} {Quantum generalizations of glauber and metropolis dynamics},\ }\href@noop {} {\bibfield  {journal} {\bibinfo  {journal} {arXiv preprint arXiv:2405.20322}\ } (\bibinfo {year} {2024})}\BibitemShut {NoStop}%
\bibitem [{\citenamefont {Jiang}\ and\ \citenamefont {Irani}(2024)}]{JiangIrani2024quantum}%
  \BibitemOpen
  \bibfield  {author} {\bibinfo {author} {\bibfnamefont {J.}~\bibnamefont {Jiang}}\ and\ \bibinfo {author} {\bibfnamefont {S.}~\bibnamefont {Irani}},\ }\bibfield  {title} {\bibinfo {title} {Quantum metropolis sampling via weak measurement},\ }\href@noop {} {\bibfield  {journal} {\bibinfo  {journal} {arXiv preprint arXiv:2406.16023}\ } (\bibinfo {year} {2024})}\BibitemShut {NoStop}%
\bibitem [{\citenamefont {Szegedy}(2004)}]{szegedy2004quantum}%
  \BibitemOpen
  \bibfield  {author} {\bibinfo {author} {\bibfnamefont {M.}~\bibnamefont {Szegedy}},\ }\bibfield  {title} {\bibinfo {title} {Quantum speed-up of markov chain based algorithms},\ }in\ \href@noop {} {\emph {\bibinfo {booktitle} {45th Annual IEEE symposium on foundations of computer science}}}\ (\bibinfo {organization} {IEEE},\ \bibinfo {year} {2004})\ pp.\ \bibinfo {pages} {32--41}\BibitemShut {NoStop}%
\bibitem [{\citenamefont {Leng}\ \emph {et~al.}(2026)\citenamefont {Leng}, \citenamefont {Ding}, \citenamefont {Chen},\ and\ \citenamefont {Lin}}]{LengDingChenLin2026operator}%
  \BibitemOpen
  \bibfield  {author} {\bibinfo {author} {\bibfnamefont {J.}~\bibnamefont {Leng}}, \bibinfo {author} {\bibfnamefont {Z.}~\bibnamefont {Ding}}, \bibinfo {author} {\bibfnamefont {Z.}~\bibnamefont {Chen}},\ and\ \bibinfo {author} {\bibfnamefont {L.}~\bibnamefont {Lin}},\ }\bibfield  {title} {\bibinfo {title} {Operator-level quantum acceleration of non-logconcave sampling},\ }\href@noop {} {\bibfield  {journal} {\bibinfo  {journal} {Proceedings of the National Academy of Sciences}\ }\textbf {\bibinfo {volume} {123}},\ \bibinfo {pages} {e2512789123} (\bibinfo {year} {2026})}\BibitemShut {NoStop}%
\bibitem [{\citenamefont {Aharonov}\ \emph {et~al.}(2010)\citenamefont {Aharonov}, \citenamefont {Arad}, \citenamefont {Landau},\ and\ \citenamefont {Vazirani}}]{AharonovAradVaziraniLandau2011detectability}%
  \BibitemOpen
  \bibfield  {author} {\bibinfo {author} {\bibfnamefont {D.}~\bibnamefont {Aharonov}}, \bibinfo {author} {\bibfnamefont {I.}~\bibnamefont {Arad}}, \bibinfo {author} {\bibfnamefont {Z.}~\bibnamefont {Landau}},\ and\ \bibinfo {author} {\bibfnamefont {U.}~\bibnamefont {Vazirani}},\ }\bibfield  {title} {\bibinfo {title} {Quantum hamiltonian complexity and the detectability lemma},\ }\href@noop {} {\bibfield  {journal} {\bibinfo  {journal} {arXiv preprint arXiv:1011.3445}\ } (\bibinfo {year} {2010})}\BibitemShut {NoStop}%
\bibitem [{\citenamefont {Arad}\ \emph {et~al.}(2013)\citenamefont {Arad}, \citenamefont {Kitaev}, \citenamefont {Landau},\ and\ \citenamefont {Vazirani}}]{AradKitaevLandauVazirani2013area}%
  \BibitemOpen
  \bibfield  {author} {\bibinfo {author} {\bibfnamefont {I.}~\bibnamefont {Arad}}, \bibinfo {author} {\bibfnamefont {A.}~\bibnamefont {Kitaev}}, \bibinfo {author} {\bibfnamefont {Z.}~\bibnamefont {Landau}},\ and\ \bibinfo {author} {\bibfnamefont {U.}~\bibnamefont {Vazirani}},\ }\bibfield  {title} {\bibinfo {title} {An area law and sub-exponential algorithm for 1d systems},\ }\href@noop {} {\bibfield  {journal} {\bibinfo  {journal} {arXiv preprint arXiv:1301.1162}\ } (\bibinfo {year} {2013})}\BibitemShut {NoStop}%
\bibitem [{\citenamefont {Aharonov}\ \emph {et~al.}(2009)\citenamefont {Aharonov}, \citenamefont {Arad}, \citenamefont {Landau},\ and\ \citenamefont {Vazirani}}]{AharonovAradLandauVazirani2009}%
  \BibitemOpen
  \bibfield  {author} {\bibinfo {author} {\bibfnamefont {D.}~\bibnamefont {Aharonov}}, \bibinfo {author} {\bibfnamefont {I.}~\bibnamefont {Arad}}, \bibinfo {author} {\bibfnamefont {Z.}~\bibnamefont {Landau}},\ and\ \bibinfo {author} {\bibfnamefont {U.}~\bibnamefont {Vazirani}},\ }\bibfield  {title} {\bibinfo {title} {The detectability lemma and quantum gap amplification},\ }in\ \href@noop {} {\emph {\bibinfo {booktitle} {Proceedings of the forty-first annual ACM symposium on Theory of computing}}}\ (\bibinfo {year} {2009})\ pp.\ \bibinfo {pages} {417--426}\BibitemShut {NoStop}%
\bibitem [{\citenamefont {Ascolani}\ \emph {et~al.}(2024)\citenamefont {Ascolani}, \citenamefont {Lavenant},\ and\ \citenamefont {Zanella}}]{ascolani2024entropy}%
  \BibitemOpen
  \bibfield  {author} {\bibinfo {author} {\bibfnamefont {F.}~\bibnamefont {Ascolani}}, \bibinfo {author} {\bibfnamefont {H.}~\bibnamefont {Lavenant}},\ and\ \bibinfo {author} {\bibfnamefont {G.}~\bibnamefont {Zanella}},\ }\bibfield  {title} {\bibinfo {title} {Entropy contraction of the {G}ibbs sampler under log-concavity},\ }\href@noop {} {\bibfield  {journal} {\bibinfo  {journal} {arXiv preprint arXiv:2410.00858}\ } (\bibinfo {year} {2024})}\BibitemShut {NoStop}%
\bibitem [{\citenamefont {Liu}\ \emph {et~al.}(1995)\citenamefont {Liu}, \citenamefont {Wong},\ and\ \citenamefont {Kong}}]{liu1995covariance}%
  \BibitemOpen
  \bibfield  {author} {\bibinfo {author} {\bibfnamefont {J.~S.}\ \bibnamefont {Liu}}, \bibinfo {author} {\bibfnamefont {W.~H.}\ \bibnamefont {Wong}},\ and\ \bibinfo {author} {\bibfnamefont {A.}~\bibnamefont {Kong}},\ }\bibfield  {title} {\bibinfo {title} {Covariance structure and convergence rate of the {G}ibbs sampler with various scans},\ }\href@noop {} {\bibfield  {journal} {\bibinfo  {journal} {Journal of the Royal Statistical Society: Series B (Methodological)}\ }\textbf {\bibinfo {volume} {57}},\ \bibinfo {pages} {157} (\bibinfo {year} {1995})}\BibitemShut {NoStop}%
\bibitem [{\citenamefont {Narayanan}\ and\ \citenamefont {Srivastava}(2022)}]{narayanan2022mixing}%
  \BibitemOpen
  \bibfield  {author} {\bibinfo {author} {\bibfnamefont {H.}~\bibnamefont {Narayanan}}\ and\ \bibinfo {author} {\bibfnamefont {P.}~\bibnamefont {Srivastava}},\ }\bibfield  {title} {\bibinfo {title} {On the mixing time of coordinate hit-and-run},\ }\href@noop {} {\bibfield  {journal} {\bibinfo  {journal} {Combinatorics, Probability and Computing}\ }\textbf {\bibinfo {volume} {31}},\ \bibinfo {pages} {320} (\bibinfo {year} {2022})}\BibitemShut {NoStop}%
\bibitem [{\citenamefont {Gilyén}\ \emph {et~al.}(2019)\citenamefont {Gilyén}, \citenamefont {Su}, \citenamefont {Low},\ and\ \citenamefont {Wiebe}}]{Gilyen_2019singularvalue}%
  \BibitemOpen
  \bibfield  {author} {\bibinfo {author} {\bibfnamefont {A.}~\bibnamefont {Gilyén}}, \bibinfo {author} {\bibfnamefont {Y.}~\bibnamefont {Su}}, \bibinfo {author} {\bibfnamefont {G.~H.}\ \bibnamefont {Low}},\ and\ \bibinfo {author} {\bibfnamefont {N.}~\bibnamefont {Wiebe}},\ }\bibfield  {title} {\bibinfo {title} {Quantum singular value transformation and beyond: exponential improvements for quantum matrix arithmetics},\ }in\ \href {https://doi.org/10.1145/3313276.3316366} {\emph {\bibinfo {booktitle} {Proceedings of the 51st Annual ACM SIGACT Symposium on Theory of Computing}}},\ \bibinfo {series and number} {STOC ’19}\ (\bibinfo  {publisher} {ACM},\ \bibinfo {year} {2019})\ p.\ \bibinfo {pages} {193–204}\BibitemShut {NoStop}%
\bibitem [{\citenamefont {Lin}\ and\ \citenamefont {Tong}(2020{\natexlab{a}})}]{LinTong2020}%
  \BibitemOpen
  \bibfield  {author} {\bibinfo {author} {\bibfnamefont {L.}~\bibnamefont {Lin}}\ and\ \bibinfo {author} {\bibfnamefont {Y.}~\bibnamefont {Tong}},\ }\bibfield  {title} {\bibinfo {title} {Near-optimal ground state preparation},\ }\href@noop {} {\bibfield  {journal} {\bibinfo  {journal} {Quantum}\ }\textbf {\bibinfo {volume} {4}},\ \bibinfo {pages} {372} (\bibinfo {year} {2020}{\natexlab{a}})}\BibitemShut {NoStop}%
\bibitem [{\citenamefont {Ge}\ \emph {et~al.}(2019)\citenamefont {Ge}, \citenamefont {Tura},\ and\ \citenamefont {Cirac}}]{Ge2018fastergroundstatepreparation}%
  \BibitemOpen
  \bibfield  {author} {\bibinfo {author} {\bibfnamefont {Y.}~\bibnamefont {Ge}}, \bibinfo {author} {\bibfnamefont {J.}~\bibnamefont {Tura}},\ and\ \bibinfo {author} {\bibfnamefont {J.~I.}\ \bibnamefont {Cirac}},\ }\bibfield  {title} {\bibinfo {title} {Faster ground state preparation and high-precision ground energy estimation with fewer qubits},\ }\href@noop {} {\bibfield  {journal} {\bibinfo  {journal} {Journal of Mathematical Physics}\ }\textbf {\bibinfo {volume} {60}} (\bibinfo {year} {2019})}\BibitemShut {NoStop}%
\bibitem [{Note1()}]{Note1}%
  \BibitemOpen
  \bibinfo {note} {Here we use $^\protect \dag $ to denote conjugation under the Hilbert-Schmidt inner product.}\BibitemShut {Stop}%
\bibitem [{\citenamefont {Anshu}\ \emph {et~al.}(2016)\citenamefont {Anshu}, \citenamefont {Arad},\ and\ \citenamefont {Vidick}}]{AnshuAradVidick2016simple}%
  \BibitemOpen
  \bibfield  {author} {\bibinfo {author} {\bibfnamefont {A.}~\bibnamefont {Anshu}}, \bibinfo {author} {\bibfnamefont {I.}~\bibnamefont {Arad}},\ and\ \bibinfo {author} {\bibfnamefont {T.}~\bibnamefont {Vidick}},\ }\bibfield  {title} {\bibinfo {title} {Simple proof of the detectability lemma and spectral gap amplification},\ }\href@noop {} {\bibfield  {journal} {\bibinfo  {journal} {Physical Review B}\ }\textbf {\bibinfo {volume} {93}},\ \bibinfo {pages} {205142} (\bibinfo {year} {2016})}\BibitemShut {NoStop}%
\bibitem [{\citenamefont {Low}\ and\ \citenamefont {Wiebe}(2018)}]{LowWiebe2018hamiltonian}%
  \BibitemOpen
  \bibfield  {author} {\bibinfo {author} {\bibfnamefont {G.~H.}\ \bibnamefont {Low}}\ and\ \bibinfo {author} {\bibfnamefont {N.}~\bibnamefont {Wiebe}},\ }\bibfield  {title} {\bibinfo {title} {Hamiltonian simulation in the interaction picture},\ }\href@noop {} {\bibfield  {journal} {\bibinfo  {journal} {arXiv preprint arXiv:1805.00675}\ } (\bibinfo {year} {2018})}\BibitemShut {NoStop}%
\bibitem [{\citenamefont {Fang}\ \emph {et~al.}(2023)\citenamefont {Fang}, \citenamefont {Lin},\ and\ \citenamefont {Tong}}]{FangLinTong2023time}%
  \BibitemOpen
  \bibfield  {author} {\bibinfo {author} {\bibfnamefont {D.}~\bibnamefont {Fang}}, \bibinfo {author} {\bibfnamefont {L.}~\bibnamefont {Lin}},\ and\ \bibinfo {author} {\bibfnamefont {Y.}~\bibnamefont {Tong}},\ }\bibfield  {title} {\bibinfo {title} {Time-marching based quantum solvers for time-dependent linear differential equations},\ }\href@noop {} {\bibfield  {journal} {\bibinfo  {journal} {Quantum}\ }\textbf {\bibinfo {volume} {7}},\ \bibinfo {pages} {955} (\bibinfo {year} {2023})}\BibitemShut {NoStop}%
\bibitem [{\citenamefont {Thibodeau}\ and\ \citenamefont {Clark}(2023)}]{ThibodeauClark2023nearly}%
  \BibitemOpen
  \bibfield  {author} {\bibinfo {author} {\bibfnamefont {M.}~\bibnamefont {Thibodeau}}\ and\ \bibinfo {author} {\bibfnamefont {B.~K.}\ \bibnamefont {Clark}},\ }\bibfield  {title} {\bibinfo {title} {Nearly-frustration-free ground state preparation},\ }\href@noop {} {\bibfield  {journal} {\bibinfo  {journal} {Quantum}\ }\textbf {\bibinfo {volume} {7}},\ \bibinfo {pages} {1084} (\bibinfo {year} {2023})}\BibitemShut {NoStop}%
\bibitem [{\citenamefont {Lin}\ and\ \citenamefont {Tong}(2020{\natexlab{b}})}]{LinTong2020optimal}%
  \BibitemOpen
  \bibfield  {author} {\bibinfo {author} {\bibfnamefont {L.}~\bibnamefont {Lin}}\ and\ \bibinfo {author} {\bibfnamefont {Y.}~\bibnamefont {Tong}},\ }\bibfield  {title} {\bibinfo {title} {Optimal polynomial based quantum eigenstate filtering with application to solving quantum linear systems},\ }\href@noop {} {\bibfield  {journal} {\bibinfo  {journal} {Quantum}\ }\textbf {\bibinfo {volume} {4}},\ \bibinfo {pages} {361} (\bibinfo {year} {2020}{\natexlab{b}})}\BibitemShut {NoStop}%
\bibitem [{Note2()}]{Note2}%
  \BibitemOpen
  \bibinfo {note} {Strictly speaking, this state is the top eigenstate of the parent Hamiltonian, but they are equivalent up to flipping sign.}\BibitemShut {Stop}%
\bibitem [{\citenamefont {Kastoryano}\ and\ \citenamefont {Brandao}(2016)}]{KastoryanoBrandao2016quantum}%
  \BibitemOpen
  \bibfield  {author} {\bibinfo {author} {\bibfnamefont {M.~J.}\ \bibnamefont {Kastoryano}}\ and\ \bibinfo {author} {\bibfnamefont {F.~G.}\ \bibnamefont {Brandao}},\ }\bibfield  {title} {\bibinfo {title} {Quantum {G}ibbs samplers: The commuting case},\ }\href@noop {} {\bibfield  {journal} {\bibinfo  {journal} {Communications in Mathematical Physics}\ }\textbf {\bibinfo {volume} {344}},\ \bibinfo {pages} {915} (\bibinfo {year} {2016})}\BibitemShut {NoStop}%
\bibitem [{\citenamefont {Hwang}\ and\ \citenamefont {Jiang}(2024)}]{HwangJiang2024gibbs}%
  \BibitemOpen
  \bibfield  {author} {\bibinfo {author} {\bibfnamefont {Y.}~\bibnamefont {Hwang}}\ and\ \bibinfo {author} {\bibfnamefont {J.}~\bibnamefont {Jiang}},\ }\bibfield  {title} {\bibinfo {title} {{G}ibbs state preparation for commuting {H}amiltonian: Mapping to classical {G}ibbs sampling},\ }\href@noop {} {\bibfield  {journal} {\bibinfo  {journal} {arXiv preprint arXiv:2410.04909}\ } (\bibinfo {year} {2024})}\BibitemShut {NoStop}%
\bibitem [{\citenamefont {Lu}\ \emph {et~al.}(2020)\citenamefont {Lu}, \citenamefont {Hsieh},\ and\ \citenamefont {Grover}}]{LuHsiehGrover2020}%
  \BibitemOpen
  \bibfield  {author} {\bibinfo {author} {\bibfnamefont {T.-C.}\ \bibnamefont {Lu}}, \bibinfo {author} {\bibfnamefont {T.~H.}\ \bibnamefont {Hsieh}},\ and\ \bibinfo {author} {\bibfnamefont {T.}~\bibnamefont {Grover}},\ }\bibfield  {title} {\bibinfo {title} {Detecting topological order at finite temperature using entanglement negativity},\ }\href {https://doi.org/10.1103/PhysRevLett.125.116801} {\bibfield  {journal} {\bibinfo  {journal} {Physical Review Letters}\ }\textbf {\bibinfo {volume} {125}},\ \bibinfo {pages} {116801} (\bibinfo {year} {2020})}\BibitemShut {NoStop}%
\bibitem [{\citenamefont {Ding}\ \emph {et~al.}(2025{\natexlab{b}})\citenamefont {Ding}, \citenamefont {Li},\ and\ \citenamefont {Lin}}]{Ding2025EfficientQuantumGibbs}%
  \BibitemOpen
  \bibfield  {author} {\bibinfo {author} {\bibfnamefont {Z.}~\bibnamefont {Ding}}, \bibinfo {author} {\bibfnamefont {B.}~\bibnamefont {Li}},\ and\ \bibinfo {author} {\bibfnamefont {L.}~\bibnamefont {Lin}},\ }\bibfield  {title} {\bibinfo {title} {Efficient quantum gibbs samplers with kubo–martin–schwinger detailed balance condition},\ }\bibfield  {journal} {\bibinfo  {journal} {Communications in Mathematical Physics}\ }\textbf {\bibinfo {volume} {406}},\ \href {https://doi.org/10.1007/s00220-025-05235-3} {10.1007/s00220-025-05235-3} (\bibinfo {year} {2025}{\natexlab{b}})\BibitemShut {NoStop}%
\bibitem [{\citenamefont {Davies}(1974)}]{Davies74}%
  \BibitemOpen
  \bibfield  {author} {\bibinfo {author} {\bibfnamefont {E.~B.}\ \bibnamefont {Davies}},\ }\bibfield  {title} {\bibinfo {title} {Markovian master equations},\ }\href@noop {} {\bibfield  {journal} {\bibinfo  {journal} {Communications in Mathematical Physics}\ }\textbf {\bibinfo {volume} {39}},\ \bibinfo {pages} {91} (\bibinfo {year} {1974})}\BibitemShut {NoStop}%
\bibitem [{\citenamefont {Davies}(1976)}]{Davies76}%
  \BibitemOpen
  \bibfield  {author} {\bibinfo {author} {\bibfnamefont {E.~B.}\ \bibnamefont {Davies}},\ }\bibfield  {title} {\bibinfo {title} {Markovian master equations. ii},\ }\href@noop {} {\bibfield  {journal} {\bibinfo  {journal} {Mathematische Annalen}\ }\textbf {\bibinfo {volume} {219}},\ \bibinfo {pages} {147} (\bibinfo {year} {1976})}\BibitemShut {NoStop}%
\bibitem [{\citenamefont {Chen}\ \emph {et~al.}(2023{\natexlab{d}})\citenamefont {Chen}, \citenamefont {Kastoryano},\ and\ \citenamefont {Gily{\'e}n}}]{ChenKastoryanoGilyen2023efficient}%
  \BibitemOpen
  \bibfield  {author} {\bibinfo {author} {\bibfnamefont {C.-F.}\ \bibnamefont {Chen}}, \bibinfo {author} {\bibfnamefont {M.~J.}\ \bibnamefont {Kastoryano}},\ and\ \bibinfo {author} {\bibfnamefont {A.}~\bibnamefont {Gily{\'e}n}},\ }\bibfield  {title} {\bibinfo {title} {An efficient and exact noncommutative quantum gibbs sampler},\ }\href@noop {} {\bibfield  {journal} {\bibinfo  {journal} {arXiv preprint arXiv:2311.09207}\ } (\bibinfo {year} {2023}{\natexlab{d}})}\BibitemShut {NoStop}%
\bibitem [{\citenamefont {Lieb}\ and\ \citenamefont {Robinson}(1972)}]{LiebRobinson1972finite}%
  \BibitemOpen
  \bibfield  {author} {\bibinfo {author} {\bibfnamefont {E.~H.}\ \bibnamefont {Lieb}}\ and\ \bibinfo {author} {\bibfnamefont {D.~W.}\ \bibnamefont {Robinson}},\ }\bibfield  {title} {\bibinfo {title} {The finite group velocity of quantum spin systems},\ }\href@noop {} {\bibfield  {journal} {\bibinfo  {journal} {Communications in Mathematical Physics}\ }\textbf {\bibinfo {volume} {28}},\ \bibinfo {pages} {251} (\bibinfo {year} {1972})}\BibitemShut {NoStop}%
\bibitem [{\citenamefont {Chen}\ \emph {et~al.}(2023{\natexlab{e}})\citenamefont {Chen}, \citenamefont {Kastoryano}, \citenamefont {Brand{\~a}o},\ and\ \citenamefont {Gily{\'e}n}}]{ChenKastoryanoBrandaoGilyen2023quantum}%
  \BibitemOpen
  \bibfield  {author} {\bibinfo {author} {\bibfnamefont {C.-F.}\ \bibnamefont {Chen}}, \bibinfo {author} {\bibfnamefont {M.~J.}\ \bibnamefont {Kastoryano}}, \bibinfo {author} {\bibfnamefont {F.~G.}\ \bibnamefont {Brand{\~a}o}},\ and\ \bibinfo {author} {\bibfnamefont {A.}~\bibnamefont {Gily{\'e}n}},\ }\bibfield  {title} {\bibinfo {title} {Quantum thermal state preparation},\ }\href@noop {} {\bibfield  {journal} {\bibinfo  {journal} {arXiv preprint arXiv:2303.18224}\ } (\bibinfo {year} {2023}{\natexlab{e}})}\BibitemShut {NoStop}%
\bibitem [{\citenamefont {Bakshi}\ \emph {et~al.}(2024)\citenamefont {Bakshi}, \citenamefont {Liu}, \citenamefont {Moitra},\ and\ \citenamefont {Tang}}]{BakshiLiuMoitraTang2024high}%
  \BibitemOpen
  \bibfield  {author} {\bibinfo {author} {\bibfnamefont {A.}~\bibnamefont {Bakshi}}, \bibinfo {author} {\bibfnamefont {A.}~\bibnamefont {Liu}}, \bibinfo {author} {\bibfnamefont {A.}~\bibnamefont {Moitra}},\ and\ \bibinfo {author} {\bibfnamefont {E.}~\bibnamefont {Tang}},\ }\bibfield  {title} {\bibinfo {title} {High-temperature {G}ibbs states are unentangled and efficiently preparable},\ }\href@noop {} {\bibfield  {journal} {\bibinfo  {journal} {arXiv preprint arXiv:2403.16850}\ } (\bibinfo {year} {2024})}\BibitemShut {NoStop}%
\bibitem [{\citenamefont {McDonough}\ \emph {et~al.}(2025)\citenamefont {McDonough}, \citenamefont {Yin}, \citenamefont {Lucas},\ and\ \citenamefont {Zhang}}]{McDonoughYinLucasZhang2025lieb}%
  \BibitemOpen
  \bibfield  {author} {\bibinfo {author} {\bibfnamefont {B.~T.}\ \bibnamefont {McDonough}}, \bibinfo {author} {\bibfnamefont {C.}~\bibnamefont {Yin}}, \bibinfo {author} {\bibfnamefont {A.}~\bibnamefont {Lucas}},\ and\ \bibinfo {author} {\bibfnamefont {C.}~\bibnamefont {Zhang}},\ }\bibfield  {title} {\bibinfo {title} {Lieb-robinson bounds with exponential-in-volume tails},\ }\href@noop {} {\bibfield  {journal} {\bibinfo  {journal} {PRX Quantum}\ }\textbf {\bibinfo {volume} {6}},\ \bibinfo {pages} {040322} (\bibinfo {year} {2025})}\BibitemShut {NoStop}%
\bibitem [{\citenamefont {Ramkumar}\ \emph {et~al.}(2025)\citenamefont {Ramkumar}, \citenamefont {Cai}, \citenamefont {Tong},\ and\ \citenamefont {Jiang}}]{RamkumarCaiTongJiang2025high}%
  \BibitemOpen
  \bibfield  {author} {\bibinfo {author} {\bibfnamefont {A.}~\bibnamefont {Ramkumar}}, \bibinfo {author} {\bibfnamefont {Y.}~\bibnamefont {Cai}}, \bibinfo {author} {\bibfnamefont {Y.}~\bibnamefont {Tong}},\ and\ \bibinfo {author} {\bibfnamefont {J.}~\bibnamefont {Jiang}},\ }\bibfield  {title} {\bibinfo {title} {High-temperature fermionic gibbs states are mixtures of gaussian states},\ }\href@noop {} {\bibfield  {journal} {\bibinfo  {journal} {arXiv preprint arXiv:2505.09730}\ } (\bibinfo {year} {2025})}\BibitemShut {NoStop}%
\bibitem [{\citenamefont {Chen}\ \emph {et~al.}(2025{\natexlab{b}})\citenamefont {Chen}, \citenamefont {Rouz{\'e}}, \citenamefont {Chen}, \citenamefont {Jiang}, \citenamefont {Scalet}, \citenamefont {Zhan}, \citenamefont {Chan}, \citenamefont {Ying},\ and\ \citenamefont {Tong}}]{ChenRouzeEtAl2025convergence}%
  \BibitemOpen
  \bibfield  {author} {\bibinfo {author} {\bibfnamefont {H.}~\bibnamefont {Chen}}, \bibinfo {author} {\bibfnamefont {C.}~\bibnamefont {Rouz{\'e}}}, \bibinfo {author} {\bibfnamefont {J.}~\bibnamefont {Chen}}, \bibinfo {author} {\bibfnamefont {J.}~\bibnamefont {Jiang}}, \bibinfo {author} {\bibfnamefont {S.~O.}\ \bibnamefont {Scalet}}, \bibinfo {author} {\bibfnamefont {Y.}~\bibnamefont {Zhan}}, \bibinfo {author} {\bibfnamefont {G.~K.}\ \bibnamefont {Chan}}, \bibinfo {author} {\bibfnamefont {L.}~\bibnamefont {Ying}},\ and\ \bibinfo {author} {\bibfnamefont {Y.}~\bibnamefont {Tong}},\ }\bibfield  {title} {\bibinfo {title} {Convergence of the cumulant expansion and polynomial-time algorithm for weakly interacting fermions},\ }\href@noop {} {\bibfield  {journal} {\bibinfo  {journal} {arXiv preprint arXiv:2512.12010}\ } (\bibinfo {year} {2025}{\natexlab{b}})}\BibitemShut {NoStop}%
\bibitem [{\citenamefont {Zlokapa}\ and\ \citenamefont {Kiani}(2026)}]{ZlokapaKiani2026syk}%
  \BibitemOpen
  \bibfield  {author} {\bibinfo {author} {\bibfnamefont {A.}~\bibnamefont {Zlokapa}}\ and\ \bibinfo {author} {\bibfnamefont {B.~T.}\ \bibnamefont {Kiani}},\ }\bibfield  {title} {\bibinfo {title} {Syk thermal expectations are classically easy at any temperature},\ }\href@noop {} {\bibfield  {journal} {\bibinfo  {journal} {arXiv preprint arXiv:2602.22619}\ } (\bibinfo {year} {2026})}\BibitemShut {NoStop}%
\bibitem [{\citenamefont {Wang}\ \emph {et~al.}(2026)\citenamefont {Wang}, \citenamefont {Zhou}, \citenamefont {Wang}, \citenamefont {Zheng}, \citenamefont {Zhang},\ and\ \citenamefont {Li}}]{WangZhouEtAl2026lindbladian}%
  \BibitemOpen
  \bibfield  {author} {\bibinfo {author} {\bibfnamefont {X.}~\bibnamefont {Wang}}, \bibinfo {author} {\bibfnamefont {S.}~\bibnamefont {Zhou}}, \bibinfo {author} {\bibfnamefont {X.}~\bibnamefont {Wang}}, \bibinfo {author} {\bibfnamefont {Y.-C.}\ \bibnamefont {Zheng}}, \bibinfo {author} {\bibfnamefont {S.}~\bibnamefont {Zhang}},\ and\ \bibinfo {author} {\bibfnamefont {T.}~\bibnamefont {Li}},\ }\bibfield  {title} {\bibinfo {title} {Lindbladian simulation with commutator bounds},\ }\href@noop {} {\bibfield  {journal} {\bibinfo  {journal} {arXiv preprint arXiv:2603.28602}\ } (\bibinfo {year} {2026})}\BibitemShut {NoStop}%
\bibitem [{\citenamefont {Haah}\ \emph {et~al.}(2021)\citenamefont {Haah}, \citenamefont {Hastings}, \citenamefont {Kothari},\ and\ \citenamefont {Low}}]{HaahHastingsKothariLow2021quantum}%
  \BibitemOpen
  \bibfield  {author} {\bibinfo {author} {\bibfnamefont {J.}~\bibnamefont {Haah}}, \bibinfo {author} {\bibfnamefont {M.~B.}\ \bibnamefont {Hastings}}, \bibinfo {author} {\bibfnamefont {R.}~\bibnamefont {Kothari}},\ and\ \bibinfo {author} {\bibfnamefont {G.~H.}\ \bibnamefont {Low}},\ }\bibfield  {title} {\bibinfo {title} {Quantum algorithm for simulating real time evolution of lattice {H}amiltonians},\ }\href@noop {} {\bibfield  {journal} {\bibinfo  {journal} {SIAM Journal on Computing}\ }\textbf {\bibinfo {volume} {52}},\ \bibinfo {pages} {FOCS18} (\bibinfo {year} {2021})}\BibitemShut {NoStop}%
\bibitem [{\citenamefont {Childs}\ \emph {et~al.}(2021)\citenamefont {Childs}, \citenamefont {Su}, \citenamefont {Tran}, \citenamefont {Wiebe},\ and\ \citenamefont {Zhu}}]{ChildsSuTranWiebeZhu2021theory}%
  \BibitemOpen
  \bibfield  {author} {\bibinfo {author} {\bibfnamefont {A.~M.}\ \bibnamefont {Childs}}, \bibinfo {author} {\bibfnamefont {Y.}~\bibnamefont {Su}}, \bibinfo {author} {\bibfnamefont {M.~C.}\ \bibnamefont {Tran}}, \bibinfo {author} {\bibfnamefont {N.}~\bibnamefont {Wiebe}},\ and\ \bibinfo {author} {\bibfnamefont {S.}~\bibnamefont {Zhu}},\ }\bibfield  {title} {\bibinfo {title} {Theory of trotter error with commutator scaling},\ }\href@noop {} {\bibfield  {journal} {\bibinfo  {journal} {Physical Review X}\ }\textbf {\bibinfo {volume} {11}},\ \bibinfo {pages} {011020} (\bibinfo {year} {2021})}\BibitemShut {NoStop}%
\end{thebibliography}%

\end{document}